\documentclass[10pt,a4paper]{article}

\usepackage[utf8x]{inputenc}
\usepackage[colorlinks=true,linkcolor=black,citecolor=black,hypertexnames=false]{hyperref}
\usepackage{amsmath}
\usepackage{amsfonts}
\usepackage{amssymb}
\usepackage{amsthm}
\usepackage{makeidx}
\usepackage{graphicx}
\usepackage{microtype}
\usepackage{tikz-cd}
\usepackage{url}
\usepackage[all]{xy}

\newtheorem{theoremintro}{Theorem}

\newtheorem{theorem}{Theorem}[section]
\newtheorem{corollary}[theorem]{Corollary}
\newtheorem{lemma}[theorem]{Lemma}
\newtheorem{proposition}[theorem]{Proposition}
\theoremstyle{definition}
\newtheorem{definition}[theorem]{Definition}
\newtheorem{remark}[theorem]{Remark}

\newcommand{\calA}{\mathcal A}
\newcommand{\calC}{\mathcal C}
\newcommand{\calZ}{\mathcal Z}
\newcommand{\Aplus}{\calA^+}
\newcommand{\Cplus}{\calC^+}
\newcommand{\Zplus}{\calZ^+}
\newcommand{\croT}{[\hspace{-0.3ex}[T]\hspace{-0.3ex}]}
\newcommand{\parT}{(\!(T)\!)}

\newcommand{\id}{\text{\rm id}}
\newcommand{\im}{\text{\rm im}\:}
\newcommand{\tr}{\text{\rm Tr}}
\newcommand{\op}{\text{\rm op}}

\newcommand{\ur}{\text{\rm ur}}
\newcommand{\Aur}{\calA^{\ur}}

\newcommand{\ev}{\text{\rm ev}}
\newcommand{\End}{\text{\rm End}}
\newcommand{\Hom}{\text{\rm Hom}}
\newcommand{\Frac}{\text{\rm Frac}}

\newcommand{\TS}{\text{\rm TS}}
\newcommand{\Ks}{K^{\text{\rm s}}}
\newcommand{\Fs}{F^{\text{\rm s}}}
\newcommand{\Asplus}{\calA^{\text{\rm s},+}}
\newcommand{\Csplus}{\calC^{\text{\rm s},+}}
\newcommand{\Zsplus}{\calZ^{\text{\rm s},+}}
\newcommand{\calAs}{\calA^{\text{\rm s}}}

\newcommand{\ord}{\text{\rm ord}}
\newcommand{\res}{\text{\rm res}}
\newcommand{\sres}{\text{\rm sres}}

\newcommand{\LRS}{\text{\rm LRS}}
\newcommand{\LG}{\text{\rm LG}}
\newcommand{\srk}{\text{\rm s-rk}}

\newcommand{\Trd}{T_\text{\rm rd}}

\begin{document}

\title{Duals of linearized Reed-Solomon codes}
\author{Xavier Caruso \& Amaury Durand}
\date{October 2021}

\maketitle

\begin{abstract}
We give a description of the duals of linearized Reed-Solomon codes in 
terms of codes obtained by taking residues of Ore rational functions.
Our construction shows in particular that, under some assumptions on
the base field, the class of linearized Reed-Solomon codes is stable
under duality.
As a byproduct of our work, we develop a theory of residues in the Ore
setting, extending the results of~\cite{caruso}.
\end{abstract}

\setcounter{tocdepth}{1}
\tableofcontents

\section*{Introduction}

One of the oldest and most basic construction of codes is due to Reed 
and Solomon and consists in evaluating polynomials of small degree in 
a large number of points; due to the decalage between the degree and 
the number of evaluation points, one can hope recovering the initial 
polynomial even when some errors occur during the evaluation, or the 
transmission of the values.
During the last decades, new problems in coding theory have emerged and 
new solutions have been proposed. In particular, one realizes than the 
rank metric (for which the distance between two codewords is given by 
the rank of some matrix) is more well-suited than the classical Hamming 
metric for some applications, \emph{e.g.} network coding~\cite{silva} 
or space-time coding~\cite{LK,robert}. The rank metric has then gained
more and more popularity over the years and many classical constructions 
have been extended to this framework.

In particular, Delsarte~\cite{delsarte}, Roth~\cite{roth} and 
Gabidulin~\cite{gabidulin} (independently)
noticed that replacing classical polynomials with linearized polynomials 
in Reed-Solomon's construction, one ends up with quite interesting 
codes as well. Those codes are nowadays referred to as Gabidulin codes.
They appear naturally as linear subspaces of matrix algebras, hence
the connexion of the rank metric.
After the work of Boucher and Ulmer~\cite{bouulm}, we prefer nowadays
working with Ore polynomials in place of linearized polynomials; this
indeed allows us to extend Gabidulin codes to arbitrary base fields,
including in particular number fields~\cite{robert} and, to some 
extent, to put the theory of Gabidulin codes in the perspective of
differential algebras~\cite{bouulm2,liu}.

More recently, Martínez-Peñas~\cite{martinez} managed to find a common 
generalization of Reed-Solomon codes, on the one hand, and Gabidulin 
codes, on the other hand.
Martínez-Peñas' codes are called \emph{linearized Reed-Solomon codes}
are involves the so-called sum-rank metric. Moreover, Martínez-Peñas
gives applications to his codes to multishot network coding
in~\cite{martinez3}. Since then, sum-rank metric codes have received
some interest (see for example~\cite{martinez3,puchinger,ravagnani,ott}).
In particular, the notion of duality for sum-rank metric codes have
been addressed in~\cite{martinez2} in which the authors proved that
the duals of certain linearized Reed-Solomon remains of the same type
(see \cite[Theorem~4]{martinez2}).

The aim of the present paper is to study in a wider generality the 
duals of linearized Reed-Solomon codes. More precisely, we will 
consider linearized Reed-Solomon codes obtained from rings of Ore 
polynomials $K[X; \theta,\delta]$ satisfying the two following
assumptions:

\medskip

\begin{tabular}{@{\hspace{1ex}}l@{\hspace{1ex}}l}
\textbf{(H1)}: &
the base ring $K$ is a \emph{commutative} field, \smallskip \\
\textbf{(H2)}: &
the subfield $F$ of $K$ consisting of elements $x 
\in K$ such that\\
& $\theta(x) = x$ and $\delta(x) = 0$ has finite 
index in $K$.
\end{tabular}

\smallskip

\noindent
These two assumptions do not allow us to work in the full generality
of Martínez-Peñas' original paper but it turns out that they are 
sufficiently weak to cover the most interesting situations.
For example, they are fulfilled when $K$ is a finite field and 
$\theta$ is (a power of) the Frobenius endomorphism or when $K$ is
the field of rational functions in the variable $t$ over a finite
field and $\delta$ is the usual derivation~$\frac d{dt}$.

In order to achieve our goal, we draw our inspiration from the 
classical case; indeed, it is a standard fact in the theory of 
algebraic geometry codes (see for instance \cite[\S 4.1.2]{tsfasman}) 
that the duals of standard Reed-Solomon codes can be described in 
terms of taking residues of differential forms over $\mathbb P^1$. 
In this paper, we extend this view point to Ore polynomials and 
linearized Reed-Solomon codes.

For this, several important ingredients are needed.
The first one is a powerful theory of residues in the Ore setting.
Such a theory has been already partially developed in a former 
paper of one of us~\cite{caruso}. Building on this work,
we extend the theory to the general setting of this article and use it to
give a new construction of codes for the sum-rank distance, that we
call \emph{linearized Goppa codes}. It turns out that these codes are
(noncanonically) isomorphic to linearized Reed-Solomon codes; however
having this alternative presentation in terms of residues will be of
crucial importance.

Two other ingredients we shall need is a notion of trace and a notion of 
duality at the level of Ore rings. The former will be given by the 
so-called \emph{reduced trace}, which is a somehow standard tool in this 
context~(see for instance \cite[\S 1.6]{jacobson-book}); however, we 
shall use a slightly unusual approach in this paper inspired by the 
theory of Azumaya algebras and better-suited to the applications we have 
in mind. As for duality maps, we introduce them in the present paper. We 
furthermore prove that the trace and the duality both satisfy 
nice commutation relations with evaluation morphisms and residue maps.
Putting all these inputs together, we finally prove that our linearized 
Goppa codes are the duals of Martínez-Peñas' linearized Reed-Solomon 
codes. As a corollary, we derive the following theorem that extends the 
theorem of Martínez-Peñas and 
Kschischang~\cite[Theorem~4]{martinez2} we have mentioned earlier.

\begin{theoremintro}
Under the assumptions~\textbf{(H1)} and~\textbf{(H2)}, the dual of
a linearized Reed-Solomon code is isomorphic to a linearized 
Reed-Solomon code.
\end{theoremintro}

The article is organized as follows.
In Section \ref{sec:LRS}, we introduce linearized Reed-Solomon codes; we 
basically follow Martínez-Peñas' treatment but reformulate it in 
a slightly different language which will help us afterwards to carry 
out our constructions.
In Section \ref{sec:trd}, we introduce the reduced trace maps, we show
that they commute with evaluation morphisms and give entirely explicit
formulas for them.
Section \ref{sec:residue} is devoted to the theory of residues: we 
define them and prove a noncommutative analogue of the residue formula.
Duality questions are discussed in Section \ref{sec:duality}: we
define a duality on Ore rings and establish useful commutation results
with evaluation morphisms and residue maps.
Finally, the construction of linearized Goppa codes and the duality
theorem is addressed in Section~\ref{sec:LG}.

\section{From Ore polynomials to codes}
\label{sec:LRS}

The aim of this section is to recall Martínez-Peñas' construction of 
linearized Reed-Solomon codes~\cite{martinez}. Actually, our presentation 
differs slightly from that of \emph{loc. cit.} in that it takes place in
a more restricted setting which allows us to use more powerful arguments 
in some places and to adopt a more conceptual view (avoiding for instance
the use of $P$-basis). This perspective on linearized Reed-Solomon codes 
will be quite useful for later developments we shall achieve in this
article.
For this reason, we have chosen to do a complete exposition of the
theory and, in particular, include full detailed proofs.

Throughout this article, we consider a field $K$ equipped with 
a ring homomorphism $\theta : K \rightarrow K$ and a
$\theta$-derivation $\delta : K \rightarrow K$, that is,
by definition, an additive mapping such that $\delta(ab) = 
\theta(a)\delta(b) + \delta(a)b$ for all $a, b \in K$.

We denote by $F$ the subfield of $K$ consisting of elements $a \in K$ 
such that $\theta(a) = a$ and $\delta(a) = 0$. \emph{We will always 
assume that the extension $K/F$ is finite}.
This hypothesis implies in particular that $\theta$ has finite order
and hence is bijective.

\subsection{Ore polynomials}
\label{ssec:Ore}

\begin{definition}
The \emph{ring of Ore polynomials} $K[X;\theta, \delta]$ is the ring 
whose elements are polynomials in $X$ over $K$ endowed with the usual 
addition and the multiplication defined by the rule:
\begin{equation*}
X a = \theta(a)X + \delta(a), \ \forall a \in K
\end{equation*}
\end{definition}

When $\theta = \id_K$ and $\delta = 0$, the ring $K[X;\theta, \delta]$
is nothing but the ring of usual univariate polynomials in $X$. In
what follows, in order to avoid this trivial cornercase, we shall
always suppose that $(\theta, \delta) \neq (\id_K, 0)$.
This additional assumption ensures in particular that $F$ is a
\emph{strict} subfield of $K$.

Although $K[X;\theta, \delta]$ is noncommutative, it shares many 
properties with the ring of usual polynomials.
First of all, we notice that the notion of degree extends 
\textit{verbatim} to Ore polynomials: if $P = \sum_i a_iX^i \in 
K[X;\theta,\delta]$, its degree is the largest integer $i$ for which $a_i \ne 0$. 
Besides, $K[X;\theta,\delta]$ is endowed with a right (resp. left\footnote{For
the left division, we use that $\theta$ is bijective.}) Euclidean 
division: if 
$A, B \in K[X;\theta,\delta]$ with $B \ne 0$, there exist unique $Q, R \in K[X;\theta,\delta]$ 
such that $A = QB+R$ (resp. $A = BQ + R$) and $\deg R < \deg B$. 
This result has the usual consequences: the noncommutative ring 
$K[X;\theta,\delta]$ is left- and right-principal, it admits GCDs and LCMs (on the 
left and on the right) and those can be computed using a noncommutative 
version of the Euclidean algorithm. In what follows, we will denote by 
$A\%B$ the remainder in the right division of $A$ by $B$.

\subsubsection*{Hilbert twist}

An important tool in the study of Ore polynomials is that of 
\emph{Hilbert twist}; this is an affine change of variables which has 
the effect of modifying the derivation. Set $\delta_0 = \theta{-}\id_K$; 
one checks that it is a $\theta$-derivation and, consequently, 
that $\delta + a \delta_0$ is also a $\theta$-derivation for all 
$a \in K$.

\begin{proposition}
\label{prop:hilbert}
For any $a \in K$, the mapping:
\begin{equation*}
K[X; \theta, \delta] 
\stackrel{\sim}{\longrightarrow} K[X; \theta, \delta{+}a \delta_0], 
\quad X \mapsto X+a
\end{equation*}
is an isomorphism of rings.
\end{proposition}

\begin{proof}
It suffices to check that 
$(X{+}a) b = \theta(b) (X{+}a) + \delta(b)$ holds
in the Ore ring $K[X; \theta, \delta{+}a \delta_0]$
for all $b \in K$, which is a simple calculation.
\end{proof}

\noindent
In addition, when $\theta$ is not the identity, one can classify 
the $\theta$-derivations of~$K$.

\begin{proposition}
\label{prop:theta-derivations}
If $\theta \neq \id_K$,
all $\theta$-derivations of $K$ are of the form $a \delta_0$
with $a \in K$.
\end{proposition}

\begin{proof}
Let $x_0 \in K$ with $\theta(x_0) \neq x_0$, \emph{i.e.}
$\delta_0(x_0) \neq 0$.
Given a $\theta$-derivation $\delta$ and $x \in K$, we write:
\begin{equation*}
\delta(x_0 x) = \theta(x_0)\delta(x) + \delta(x_0)x 
= \theta(x)\delta(x_0) + \delta(x)x_0
\end{equation*}
for what we deduce that
$\delta(x) = \frac{\delta(x_0)}{\delta_0(x_0)}\delta_0(x)$
and finally that $\delta$ is proportionnal to $\delta_0$.
\end{proof}

Combining Propositions~\ref{prop:hilbert}
and~\ref{prop:theta-derivations}, we find that the Ore polynomial
ring $K[X;\theta,\delta]$ is isomorphic to $K[X;\theta,0]$ as
soon as $\theta$ is not the identity. Therefore, one can split
the study of Ore polynomials over fields into two cases: the 
``endomorphic'' one where $\delta = 0$ and the ``differential''
one where $\theta = \id_K$.

\subsubsection*{The centre}

Recall that the centre of a noncommutative ring $A$ is by definition 
the subset of $A$ consisting of elements $x$ such that $xy=yx$ for all 
$y \in A$; in particular, the centre is always commutative.

It turns out that the centre of Ore polynomial rings plays a quite 
important role and can be explicitely determined.
Precisely, when $\delta = 0$, one checks that the centre of
$K[X;\theta,0]$ is $F[X^s]$ where $s$ is the order of $\theta$
(and we recall that $F$ is by definition the subfield of $K$ fixed
by $\theta$). The case where $\theta \neq \id_K$ reduces to
the previous one using Hilbert twist;
indeed, from Proposition~\ref{prop:theta-derivations}, we know
that $\delta = a \delta_0$ for some $a \in K$ and it then follows
from Proposition~\ref{prop:hilbert} that the centre of 
$K[X;\theta,\delta]$ is $F[(X{+}a)^s]$ where, again, $s$ denotes
the order of $\theta$.

The case where $\theta = \id_K$ is, by far, the more difficult
one. By chance, it has already been studied in details 
in~\cite{arriagada}.
Recall that, in this case, $F$ is defined as the subfield of
constants of $K$.
Let $Z(X) \in K[X]$ be the monic polynomial of minimal degree
annihilating $\delta$ (which exists because all the $\delta^i$
are $F$-linear mappings of $K$, which is finite dimensional over
$F$).
The centre of $K[X;\id_K, \delta]$ is then the subring
$F[Z(X)]$. Besides, it is proved in~\emph{loc. cit.} that $Z(X)$
is a monic linearized polynomial with coefficients in $F$, 
\emph{i.e.} it takes the form:
\begin{equation}
\label{eq:ZX}
Z(X) = X^{p^r} + z_{r-1} X^{p^{r-1}} + \cdots + 
z_1 X^p + z_0 X \quad (z_i \in F)
\end{equation}
and its degree $p^r$ is the degree of the extension $K/F$.

To summarize, the following proposition holds in all cases.

\begin{proposition}
\label{prop:centre}
There exists a monic Ore polynomial $Z(X) \in K[X;\theta,\delta]$ 
such that the centre of $K[X;\theta,\delta]$ is $F[Z(X)]$.
Moreover $\deg Z(X) = [K:F]$.
\end{proposition}

\begin{remark}
The equality
$\sum_{i=0}^d a_iZ(X)^i = \sum_{i=0}^e b_iZ(X)^i$
readily implies that $d=e$ and $a_i = b_i$ for all $i$.
As a consequence, the centre $F[Z(X)]$ is an actual polynomial 
ring in one variable with coefficients in $F$.
\end{remark}

\begin{remark}
\label{rem:ZX}
The condition of Proposition~\ref{prop:centre} only
determines $Z(X)$ uniquely but only up to an additive constant in 
$F$. However, one can always normalize it by requiring it to be
a linearized polynomial when $\theta = \id_K$, or a power of
a Ore polynomial of degree $1$ otherwise.
\end{remark}

\subsection{On the evaluation of Ore polynomials}
\label{ssec:evalOre}

Evaluating Ore polynomials is not straightforward; indeed 
performing the substitution $X \mapsto c$ for some $c \in K$
does not define a ring homomorphism and hence is not relevant.
An option which is often considered (see for instance \cite{lam})
is to define $P(c)$ as the remainder in the division of $P$ by $X{-}c$.
However, in this article, we will follow a different path 
based on the notion of pseudo-linear morphism which was first
introduced by Jacobson in~\cite{jacobson2} and then further
developed by Leroy in~\cite{leroy}.

\subsubsection{Definition of evaluation maps}

We start by recalling Jacobson's definition of pseudo-linear
morphisms.

\begin{definition}
Let $M$ be a vector space over $K$.
A \textit{pseudo-linear endomorphism} $u : M \rightarrow M$ 
(with respect to $\theta$ and $\delta$) is an additive map verifying:
$$u(ax) = \theta(a)u(x) + \delta(a)x$$
for all $a \in K$ and $x \in M$.
\end{definition}

We observe that any pseudo-linear morphism is \textit{a fortiori} 
$F$-linear.
If $u : M \rightarrow M$ is a pseudo-linear morphism and $P
= \sum_{i} a_iX^i \in K[X;\theta,\delta]$ is a Ore polynomial, we
define $P(u) = \sum_i a_iu^i$. A simple computation then shows that
$P(u)\circ Q(u) = (PQ)(u)$ for all $P, Q \in K[X;\theta,\delta]$.
In other words, the mapping:
\begin{equation*}
\ev_u : K[X;\theta,\delta] \longrightarrow \End_F(M), \quad P(X) \mapsto P(u)
\end{equation*}
is a ring homomorphism (where $\End_F(M)$ denotes the ring of 
$F$-linear endomorphisms of $M$).
The case where $M$ is $K$ itself deserves particular attention.
Indeed, we first observe that evaluation is then closely related to 
Euclidean division thanks to the formula:
\begin{equation}
\label{eq:eveuc}
\textstyle \ev_u(P)(a) = (aP)\,\%\left(X-\frac{u(a)}{a}\right)
\end{equation}
which is correct for any pseudo-linear endomorphism $u$ of $K$, any $P 
\in K[X;\theta,\delta]$ and any $a \in K$, $a \neq 0$ (see also
\cite[Theorem~2.8]{leroy}). Second, we have a 
complete classification of pseudo-linear endomorphisms of $K$.

\begin{proposition}
The pseudo-linear endomorphisms of $K$ are exactly the maps of the form $\delta + c\theta$ with $c \in K$.
\end{proposition}

\begin{proof}
It is straightforward to check that $\delta + c \theta$ is a 
pseudo-linear morphism for any $c \in K$.
Conversely, let $u : K \to K$ be a pseudo-linear endomorphism.
For $x \in K$, it follows from the definition that $u(x)
= \theta(x) u(1) + \delta(x)$. Therefore $u = \delta + c \theta$
with $c = u(1)$.
\end{proof}

In what follows, we will often use the notation $\ev_c$ in place of 
$\ev_{\delta+c\theta}$.
We notice that those evaluation maps are compatible with Hilbert twists
in the sense that the diagram below commutes for all $a, c \in K$:
$$\xymatrix @C=5em {
K[X;\theta,\delta] \ar[r]^-{X \mapsto X+a} \ar[d]_-{\ev_c} &
K[X;\theta,\delta{+}a\delta_0] \ar[d]^-{\ev_{c-a}} \\
\End_F(K) \ar@{=}[r] & \End_F(K)}$$

\subsubsection{The kernel of the evaluation maps}

We recall from Proposition~\ref{prop:centre}
that the centre of $K[X;\theta,\delta]$ is 
of the form $F[Z(X)]$ and that the Ore polynomial $Z(X)$ can be
normalized using the additional conditions of Remark~\ref{rem:ZX}.
For $c \in K$, we define $\upsilon(c)$ as the remainder of the right Euclidean
division of $Z(X)$ by $X{-}c$; by definition, $Z(X){-}\upsilon(c)$ is then
a right multiple of $X{-}c$.

\begin{definition}
\label{def:ramified}
We say that an element $c \in K$ is \emph{ramified} (with respect to 
$\theta$ and $\delta$) if $\delta + c\theta$ is a scalar multiple of 
$\id_K$.
Otherwise, we say that $c$ is \emph{unramified}.
\end{definition}

One checks that 
$K$ contains at most one ramified element.
Precisely, when $\theta = \id_K$ (and $\delta \neq 0$), all
elements of $K$ are unramified while, when $\theta \neq \id_K$ 
and $\delta = a \delta_0$, the unique ramified element of $K$ is
$-a$.

\begin{proposition}
\label{prop:kerevc}
Let $c$ be an unramified element of $K$.
Then $\ev_c : K[X;\theta,\delta] \to \End_F(K)$ is surjective
and its kernel is the principal left ideal generated by 
$Z(X)-\upsilon(c)$.
\end{proposition}

\begin{proof}
Let us first assume that $\delta = 0$ and let $n$ be the order
of $\theta$. Then $Z(X) = X^n$ and $n = [K:F]$. Moreover, by 
Artin's linear independence theorem, we know that the family
$(\id_K, \theta, \ldots, \theta^{n-1})$ is $K$-free in $\End_F(K)$.
By comparing dimensions, it is then enough to prove that $\ev_c$
vanishes on $X^n - \upsilon(c)$. Write $\varphi = \ev_c\big(X^n - \upsilon(c)\big)
= (c\theta)^n - \upsilon(c) \:\id_K$.
On the one hand, a direct computation using that $\theta^n = \id_K$
indicates that $\varphi$ must be a multiple of $\id_K$.
On the other hand, the fact that $X^n - \upsilon(c)$ is a right multiple of 
$X-c$ implies that $\varphi$ vanishes at $1$. Putting these two
inputs together, we deduce that $\varphi$ vanishes, as wanted.
The case where $\theta \neq \id_K$ reduces to the previous one using 
a well-chosen Hilbert twist (see Propositions~\ref{prop:hilbert}
and~\ref{prop:theta-derivations}).

Finally, we suppose that $\theta = \id_K$. 
Recall that $Z(X)$ is defined in this case as the 
minimal polynomial of $\delta$ and that it is linearized polynomial 
with coefficients in $F$. Observe now that
\begin{equation}
\label{eq:ZXc}
Z(X{-}c) = Z(X) - \upsilon(c).
\end{equation}
Indeed, it follows from Proposition~\ref{prop:hilbert} that the
mapping $X \mapsto X{-}c$ induces an automorphism of 
$K[X;\id_K,\delta]$, implying that $Z(X{-}c)$ has to be central
and hence of the form $Z(X) - a$ for some $a \in F$. Taking
remainders modulo $X{-}c$, we finally find $a = \upsilon(c)$.

It follows from Eq.~\eqref{eq:ZXc} that $Z(X) - \upsilon(c)$ lies in the 
kernel of $\ev_c$. Moreover, the fact that $Z(X)$ is the minimal 
polynomial of $\delta$ says that the family $(\delta^i)_{0 \leq i < 
p^r}$ is linearly independent over $K$. Hence
the family $((\delta+c{\cdot}\id_K)^i)_{0 \leq i < p^r}$ is also,
implying eventually that the kernel of $\ev_c$ is exactly the ideal
generated by $Z(X) - \upsilon(c)$. The surjectivity of $\ev_c$ follows by
comparing dimensions over $F$ (and using that $\deg Z(X) = [K:F]$).
\end{proof}

An interesting (and unexpected) corollary of 
Proposition~\ref{prop:kerevc} is the following.

\begin{corollary}
\label{cor:imnu}
For all $c \in K$, we have $\upsilon(c) \in F$.
\end{corollary}

\begin{proof}
To simplify notations, we write $a = \upsilon(c)$.
From Proposition~\ref{prop:kerevc}, we derive that the left
principal ideal generated by $Z(X){-}a$ is actually two-sided
because it appears as the kernel of a ring homomorphism.
In particular, it contains the commutator of $Z(X){-}a$ and
$X$, which is $(\theta(a) - a) X + \delta(a)$.
Therefore the latter Ore polynomial is right-divisible by 
$Z(X){-}a$ and so, by comparing degrees, it must vanish.
Hence $\theta(a) = a$ and $\delta(a) = 0$, which exactly means
that $a \in F$.
\end{proof}

Beyond Corollary~\ref{cor:imnu}, it is possible to write down 
explicit formulas for $\upsilon(c)$.
Concretely, when $\delta = 0$, it is easy to check that the
$\upsilon = N_{K/F}$, the norm of $K$ over $F$. As usual,
the case where $\delta = a \delta_0$ reduces to the previous one 
using Hilbert twist; in this setting, we find $\upsilon(c) = N_{K/F}(a+c)$.
Finally, when $\theta = \id_K$, it follows from the computations of
\cite[No~12]{jacobson} (see in particular Eq.~(35)) that
\begin{equation}
\label{eq:upsilon}
\upsilon(c) = 
\sum_{i=0}^r \sum_{j=0}^i \left(z_i\delta^{p^j-1}(c) \right)^{p^{i-j}}
\end{equation}
where the $z_i$'s are the coefficients of $Z(X)$ as in
Eq.~\eqref{eq:ZX}.
These explicit descriptions provide an alternative proof of
Corollary~\ref{cor:imnu}.

\subsubsection{Zeros of Ore polynomials}

In the standard commutative case, it is well known that the number of roots of a polynomial cannot exceed its degree. In the Ore setting, analoguous bounds exist.

\begin{proposition}
\label{prop:bounddeg}
Let $c$ be an unramified element of $K$ and
let $P \in K[X;\theta,\delta]$ be a nonzero polynomial. We have
$$\dim_F \ker(\ev_c(P)) \leq \deg P$$
and equality holds if and only if $P$ divides $Z(X) - \upsilon(c)$.
\end{proposition}

\begin{proof}
Let $I$ be the left ideal of $\End_F(K)$ consisting of linear morphisms 
vanishing on $\ker (\ev_c(P))$. The inverse image of $I$ by $\ev_c$ is 
an ideal of $K[X;\theta,\delta]$. Since $K[X;\theta,\delta]$ is 
principal, there exists $Q \in K[X;\theta,\delta]$ such that 
$\ev_c^{-1}(I) = K[X;\theta,\delta] Q$. We have the following 
equalities of dimensions:
\begin{align*}
r \cdot \deg Q & = \dim_F K[X;\theta,\delta] /K[X;\theta,\delta] Q \\
 & = \dim_F K[X;\theta,\delta]/\ev_c^{-1}(I) \\
 & = \dim_F \End_F(K)/I  \\
 & = r \cdot \dim_F \ker(\ev_c(P)).
\end{align*}
We deduce that $\deg Q = \dim_F \ker(\ev_c(P))$. Since $P$ obviously 
belongs to $\ev_c^{-1}(I)$, we conclude that $P$ divides $Q$, showing 
the inequality of the proposition.

Moreover, equality holds if and only if $\deg P = \deg Q$, \emph{i.e.}
$P$ divides $Q$. If $P$ indeed divides $Q$, we conclude that $P$ divides
$Z(X) - \upsilon(c)$ because $Q$ is a divisor of $Z(X) - \upsilon(c)$
thanks to Proposition~\ref{prop:kerevc}.
Conversely, if $P$ divides $Z(X)-\upsilon(c)$, we write 
$Z(X)-\upsilon(c) = PP'$ and, applying the inequality to $P$ and $P'$, 
we get:
$$\dim_F \ker(\ev_c(P)) + \dim_F \ker(\ev_c(P')) 
\leq \deg P + \deg P' = r.$$
On the other hand, we have:
$$r = \dim_F \ker(\ev_c(PP')) \leq
\dim_F \ker(\ev_c(P)) + \dim_F \ker(\ev_c(P')).$$
All inequalities then need to be equalities, which concludes the
proof.
\end{proof}

\begin{corollary}
\label{cor:vanishpol}
Let $c$ be an unramified element of $K$.
Given a $F$-linear subspace $V$ of $K$, there exists a unique monic 
polynomial $P \in K[X;\theta,\delta]$ such that $\ker (\ev_c(P)) = V$ 
and $\deg P = \dim_F V$.
\end{corollary}

\begin{proof}
It suffices to take the Ore polynomial $P$ defined by
$\ev_c^{-1}(I) = K[X;\theta,\delta] P$ where $I$ denotes the
ideal of $\End_F(K)$ consisting of functions vanishing on $V$.
\end{proof}

We can extend the above results to multiple evaluations.

\begin{theorem}
\label{theo:multipleeval}
Let $c_1, \ldots, c_m$ be unramified elements of $K$ such that
the $\upsilon(c_i)'s$ are pairwise distinct. 
\begin{enumerate}
\item For all $P \in K[X;\theta,\delta]$, $P \ne 0$ :
\begin{equation*}
\sum_{i=1}^m \dim_F \ker(\ev_{c_i}(P)) \leq \deg P
\end{equation*}
and equality holds if and only if $P$ divides $\prod_{i=1}^m Z(X)-\upsilon(c_i)$.
\item Given $F$-linear subspaces $V_1, \dots, V_m$ of $K$, there exists 
a unique monic polynomial $P \in K[X;\theta,\delta]$ of degree
$\sum_{i=1}^m \dim_F V_i$ such that $\ev_{c_i}(P)$ vanishes on $V_i$.
\end{enumerate}
\end{theorem}
\begin{proof}
We consider the ring homomorphism:
$$\varepsilon : K[X;\theta,\delta] \rightarrow \End_F(K)^m, \quad 
P \mapsto \big(\ev_{c_1}(P), \dots, \ev_{c_m}(P)\big).$$
Let $P \in \ker \varepsilon$. By Proposition~\ref{prop:kerevc}, $P$
is a multiple of $Z(X)-\upsilon(c_i)$ for all $i$ . Since these polynomials 
are pairwise coprime, we deduce that $P$ is divisible by 
$\prod_{i=1}^m Z(X)-\upsilon(c_i)$. By comparing dimensions, we find that
$\varepsilon$ is surjective and its kernel is the principal ideal
generated by $\prod_{i=1}^m Z(X)-\upsilon(c_i)$.
With this imput, the proof is similar to the proofs of 
Proposition~\ref{prop:bounddeg} and Corollary~\ref{cor:vanishpol}.
\end{proof}

\subsection{Linearised Reed-Solomon codes}
\label{ssec:LRS}

We fix a positive integer $m$ together with a tuple $\underline{V} = 
(V_1, \dots, V_m)$ of $F$-linear subspaces of $K$. 
We set:
$$\Hom_F(\underline V, K) = 
\Hom_F(V_1,K) \times \cdots \times \Hom_F(V_m, K).$$
It is a vector space over $K$ of dimension
$\sum_{i=1}^m \dim_F V_i$. Following~\cite[Definition~25]{martinez},
we equip $\Hom_F(\underline V, K)$ with the sum-rank distance 
defined as follows.

\begin{definition}
The \emph{sum-rank weight} of
$\underline \varphi = \big(\varphi_1, \dots, \varphi_m\big) \in
\Hom_F(\underline V, K)$~is:
$$w_{\srk}(\underline{\varphi}) = \sum_{i=1}^m \dim_F \varphi_i(V_i)$$
The \emph{sum-rank distance} between
$\underline \varphi$ and $\underline \psi$ is
$d_\srk(\underline \varphi, \underline \psi) =
w_\srk(\underline \varphi - \underline \psi)$.
\end{definition}

Throughout this article, we will use the word \emph{code} to refer 
to a $K$-linear subspace of $\Hom_F(\underline V, K)$. By definition, 
the \emph{length} of a code $C$ sitting in $\Hom_F(\underline V, K)$ is 
$\sum_{i=1}^m \dim_F V_i$, 
its \emph{dimension} is $\dim_K C$ and
its \emph{minimal distance} is the minimal sum-rank weight of a 
nonzero element of $C$.

These three parameters are related by an analogue of the classical 
Singleton bound, which reads $k + d \leq n + 1$ in our setting
(see~\cite[Proposition~34]{martinez}).
Codes attaining this bound are called MSRD (for \emph{Maximal
Sum-Rank Distance}).

\begin{definition}
\label{def:LRS}
Let $k$ be an integer and 
$\underline{c} = (c_1, \dots, c_m)$ be a tuple of
unramified elements of $K$. We set:
$$\begin{array}{rcl}
\ev_{\underline c, \underline V} : 
K[X;\theta,\delta] & \rightarrow & \Hom_F(\underline V,K) \smallskip \\
P & \mapsto & 
  \big(\ev_{c_1}(P)_{|V_1},\, \ldots,\, \ev_{c_m}(P)_{|V_m}\big).
\end{array}$$
The \textit{linearised Reed-Solomon code} associated to $(k, 
\underline c, \underline V)$ is:
$$\LRS(k, \underline{c}, \underline{V}) = 
\ev_{\underline c, \underline V}\big(K[X;\theta,\delta]_{<k}\big)$$
where, by definition, $K[X;\theta,\delta]_{<k}$ denotes the subset of 
$K[X;\theta,\delta]$ of Ore polynomial of degree strictly less than
$k$.
\end{definition}

\begin{remark}
The linearized Reed-Solomon codes appear as a common generalization of 
Reed-Solomon codes on the one hand and Gabidulin codes on the other 
hand. Indeed, if we had been working with classical polynomials instead 
of Ore polynomials, we would have end up with a usual Reed-Solomon code,
while the case $m = 1$ reduces to Gabidulin codes.
\end{remark}

The following theorem gives the parameters of the linearized 
Reed-Solomon codes.

\begin{theorem}\label{theo::caracRSL}
Let $k$ be an integer, $\underline c = (c_1, \ldots, c_m)$ be
a tuple of unramified elements of $K$ and 
$\underline V = (V_1, \ldots, V_m)$ be a tuple of $F$-linear
subspaces of $K$. We set $n = \sum_{i=1}^m \dim_F V_i$.

If $k \leq n$ and the $\upsilon(c_i)$'s are pairwise distinct, 
the code $\LRS(k, \underline{c}, \underline{V})$ has length~$n$,
dimension~$k$ and minimal distance~$n-k+1$; in particular, it is
MSRD.
\end{theorem}

\begin{proof}
The fact that $\LRS(k, \underline{c}, \underline{V})$ has length~$n$
is obvious.
Let $P \in K[X;\theta,\delta]$ be a Ore polynomial of degree strictly
less than $k$. For $i \in \{1, \ldots, m\}$, set $f_i = \ev_{c_i}(P)$ 
and let $\varphi_i: V_i \to K$ be the restriction of $f_i$ to $V_i$.
From Theorem~\ref{theo:multipleeval}, we derive:
$$\sum_{i=1}^m \dim_F \ker \varphi_i \leq
\sum_{i=1}^m \dim_F \ker f_i \leq \deg P < k$$
the first inequality coming from the obvious inclusion $\ker \varphi_i
\subset \ker f_i$. By the rank-nullity theorem, we conclude that 
$w_{\srk}(\varphi_1, \ldots, \varphi_m) > n-k$, which concludes
the proof.
\end{proof}

\section{Reduced trace of Ore polynomials}
\label{sec:trd}

The aim of this section is to introduce and prove the main properties 
of the reduced trace map over the rings of Ore polynomials. This notion 
can be seen as an analogue of the usual trace over ring of matrices and, 
for this reason, it will play a central role when we will study duality 
in \S \ref{sec:duality}.

We keep the notations and hypothesis of \S \ref{sec:LRS}: the letter $K$ 
denotes a field equipped with an automorphism $\theta : K \to K$ and a 
$\theta$-derivation $\delta : K \to K$. We let $F$ be the subfield of 
$K$ consisting of elements $x$ such that $\theta(x) = x$ and $\delta(x) 
= 0$ and assume that the extension $K/F$ is finite.

In order to simplify notations, we set $\Aplus = K[X;\theta, \delta]$ 
and let $\Zplus$ be its centre. By Proposition~\ref{prop:centre}, we 
know that $\Zplus = F[Z(X)]$ for some Ore polynomial $Z(X) \in \Aplus$ 
of degree $s = [K:F]$. Moreover, $Z(X)$ can be chosen in a canonical way 
(see Remark~\ref{rem:ZX}).
We set in addition $\Cplus = K[Z(X)]$; it is a commutative subring
of $\Aplus$ containing $\Zplus$. Besides, $\Aplus$ appreas as a free
left-module of rank $s$ over $\Cplus$, a basis of it being 
given by $(1, X, \ldots, X^{s-1})$. We will refer to it as the
canonical basis of $\Aplus$ over $\Cplus$.

\begin{definition}
\label{def:trd}
Let $f \in \Aplus$. The \emph{reduced trace} of $f$, denoted by 
$\Trd(f)$, is the trace of the map $x \mapsto xf$ viewed as a 
$\Cplus$-linear endomorphism of $\Aplus$.
\end{definition}

\noindent
This construction defines a mapping $\Trd : \Aplus \to \Cplus$.
The reader should pay attention to the fact that $\Trd$ is \emph{not} 
$\Cplus$-linear because of noncommutativity; it is however 
$\Zplus$-linear.
Besides, it satisfies the classical trace relation $\Trd(fg) = 
\Trd(gf)$ for all $f, g \in \Aplus$. Another remarkable property
of $\Trd$ is that it assumes values in $\Zplus$; this result is not
obvious from Definition~\ref{def:trd} but is a consequence
of the explicit formulas we will obtain in \S \ref{ssec:explicittrd}
(see Propositions~\ref{prop:trdendo} and~\ref{prop:trddiff}).

\subsection{An explicit formula}
\label{ssec:explicittrd}

Although Definition~\ref{def:trd} is already rather explicit, it
is possible to simplify it further and end up with simple close
expressions for the reduced trace. The objective of this subsection 
is to derive such formulas.
Our first theorem in this direction addresses the case $\delta = 0$, 
which is the simplest one. In this situation, we recall that $F$ is 
the fixed subfield of $\theta$ and $Z(X) = X^s$. The integer $s$ is
exactly the order of $\theta$.

\begin{proposition}
\label{prop:trdendo}
We assume that $\delta = 0$. For $f = \sum_i a_i X^i \in \Aplus$,
we have:
$$\Trd(f) = \sum_i \tr_{K/F}(a_{si}) X^{si} =
\sum_i \tr_{K/F}(a_{si}) Z(X)^i$$
where $\tr_{K/F}$ is the trace map of $K$ over $F$.
\end{proposition}

\begin{proof}
By $\Zplus$-linearity, it is enough to prove that 
$\Trd(a) = \tr_{K/F}(a)$ and $\Trd(a X^i) = 0$ for $a \in K$
and $i \in \{1, \ldots, s{-}1\}$.
For the first assertion, we observe that the matrix of the 
multiplication map $x \mapsto xa$ in the canonical basis is
diagonal and its diagonal entries are $a, \theta(a), \ldots,
\theta^{s-1}(a)$. Therefore its trace is $\tr_{K/F}(a)$.
Similarly, when $1 \leq i < s$, the matrix of the multiplication 
map $x \mapsto x aX^i$ has only nonzero entries at the position
$(u,v)$ with $v \equiv u + i \pmod s$.
In particular, its diagonal vanishes. Hence so does its trace.
\end{proof}

As usual, the case where $\theta \neq \id_K$ reduces to the previous
one using Hilbert twists. Precisely, writing $\delta = a \delta_0$
(see Proposition~\ref{prop:theta-derivations}), we find the formula:
$$\Trd\Big(\sum_i a_i (X{+}a)^i\Big) = 
\sum_i \tr_{K/F}(a_{si}) (X{+}a)^{si}.$$
We now come to the case where $\theta = \id_K$. In this situation,
we recall that $F$ is the field of constants of $K$ and that
$Z(X)$ is defined as the minimal polynomial of $\delta$; besides,
we know that $Z(X)$ is a linearized polynomial over $F$, \emph{i.e.} 
it takes the form:
\begin{equation}
\label{eq:ZX2}
Z(X) = X^{p^r} + z_{r-1} X^{p^{r-1}} + \cdots + 
z_1 X^p + z_0 X
\end{equation}
with all the coefficients $z_i$ in $F$. We then have $s = p^r$.
For convenience, we also set $z_r = 1$.

\begin{proposition}
\label{prop:trddiff}
For $a \in K$, we have:

\vspace{1em}

\noindent
$\begin{array}{@{}c@{\hspace{5em}}l}
& \Trd(a) = \Trd(aX) = \cdots = \Trd\big(aX^{p^r-2}\big) = 0 \smallskip \\
\text{and:}
& \Trd\big(aX^{p^r-1}\big) = \displaystyle \sum_{j=0}^r z_j \delta^{p^j-1}(a).
\end{array}$
\end{proposition}

\begin{remark}
Thanks to $\Zplus$-linearity, the above formulas
are enough to compute the reduced trace of any Ore polynomial
$f \in \Aplus$. Besides, one immediately checks that the quantity
$\sum_{j=0}^r z_j \delta^{p^j-1}(a)$ is annihilated by $\delta$
and hence lies in~$F$. It follows from these observations that
$\Trd$ takes its values in $\Zplus = F[Z(X)]$, in accordance
with what we have announced in the introduction of \S \ref{sec:trd}.
\end{remark}

The rest of this subsection is devoted to the proof of 
Proposition~\ref{prop:trddiff}.
For a Ore polynomial $P$, we denote by $\pi_i(P)$ its coefficient in 
front of $X^i$ when $P$ is written in the canonical $\Cplus$-basis 
$\{1, X, \dots, X^{p^r-1}\}$ of $\Aplus$.
This defines a $\Cplus$-linear map $\pi_i : \Aplus \to \Cplus$.

\begin{lemma}\label{lem:trace reduite}
For $0 \leq i < p^r$ and $-i \leq v < p^r$, we have:
$$\begin{array}{r@{\hspace{0.5ex}}l@{\qquad}l}
\pi_i(X^{i+v}) & = 1 & \text{if } v = 0 \\
& = -z_j & \text{if } v = p^r-p^j \text{ and } i \geq p^j \\
& = 0 & \text{otherwise.}
\end{array}$$
\end{lemma}

\begin{proof}
We begin by noticing that, under the condition $v = p^r - p^j$, the
fact that $i \geq p^j$ is equivalent to $i + v \geq p^r$.
If $i+v < p^r$ then $X^{i+v}$ is already written in the canonical basis
and the result follows. On the contrary, if $i + v \geq p^r$, we write
$i+v = p^r + k$ and:
\begin{equation}
\label{eq:Xiv}
X^{i+v} = X^{p^r}X^k =
Z(X)X^k - \sum_{j=0}^{r-1} z_jX^{p^j+k}.
\end{equation}
If $i+v < 2p^r -p^{r-1}$, then all the exposants $p^j+k$ are less than
$p^r$ and the writing above is the decomposition of $X^{i+v}$ in the 
canonical basis. Hence the lemma follows in this case. 

Finally, when $i + v \geq 2p^r -p^{r-1}$, applying $\pi_i$ to
Eq.~\eqref{eq:Xiv}, we get:
$$\pi_i(X^{i+v}) = Z(X) \pi_i(X^k) - \sum_{j=0}^{r-1} z_j \pi_i(X^{p^j+k}).$$
Observing that
$p^j + k \leq p^{r-1} + k < p^{r+1} + p^r \leq 2p^r - p^{r-1}$
whenever $0 \leq j < r$, the lemma follows from what we have done
previously.
\end{proof}

Let $n \in \{0, \ldots, p^r{-}1\}$.
By definition, the reduced trace of $a X^n$ is given by:
$$\Trd(aX^n) = \sum_{i=0}^{p^r-1} \pi_i(X^iaX^n).$$
Moving all the $X$ to the right and writing $v = n-j$, we obtain
the formula:
\begin{align*}
\Trd(aX^n) 
 & = \sum_{i=0}^{p^r-1} \sum_{v=n-i}^n \binom{i}{n{-}v} \delta^{n-v}(a) \cdot \pi_i(X^{i+v}) \\
 & = \sum_{i=0}^{p^r-1} \,\, \sum_{v=-i}^{p^r-1} \, \binom{i}{n{-}v} \delta^{n-v}(a) \cdot \pi_i(X^{i+v})
\end{align*}
the last equality being true because the binomial coefficients 
$\binom i {n-v}$ vanish when $v$ is outside the range $[n{-}i, n]$.
It follows from Lemma~\ref{lem:trace reduite} that only the terms with
$v = 0$ and $v = p^r - p^j$ contribute to the sum.
Precisely, the contribution of the summands corresponding to $v = 0$ 
is:
$$\begin{array}{r@{\hspace{0.5ex}}ll}
\displaystyle
\sum_{i=0}^{p^r-1} \binom{i}{n} \delta^n(a) = \binom{p^r}{n{+}1} \delta^n(a)
 & = 0 & \text{if } n < p^r - 1 \\
 & = \delta^{p^r-1}(a) & \text{if } n = p^r - 1.
\end{array}$$
Similarly,
the contribution of the summands coming from $v = p^r - p^j$ is:
$$C_j = -z_j \sum_{i=p^j}^{p^r-1}  \binom{i}{n{-}v} \delta^{n-v}(a) = 
-z_j\left(\binom{p^r}{n{-}v{+}1} - \binom{p^j}{n{-}v{+}1}\right) \delta^{n-v}(a).$$
When $n < p^r - 1$, \emph{i.e.} $n - v + 1 < p^j$, the binomial
coefficients $\binom{p^r}{n-v+1}$ and $\binom{p^j}{n-v+1}$ both vanish,
implying that $C_j = 0$ as well. On the contrary, when $n = p^r - 1$, we
find $C_j = z_j \delta^{p^j-1}(a)$.

Putting all the contributions together, we finally end up with the 
formula of Proposition~\ref{prop:trddiff}.

\subsection{Reduced trace and evaluation}
\label{ssec:trdeval}

We recall that we have introduced evaluation maps in \S
\ref{ssec:evalOre}. Precisely, for all unramified elements $c$ of 
$K$, we have defined a ring homomorphism $\ev_c : \Aplus \to \End_F(K)$
taking a Ore polynomial $f(X)$ to $f(\delta + c\theta)$.
We also recall that we have defined $\upsilon(c)$ as the remainder
of the right Euclidean division of $Z(X)$ by $X{-}c$. 
By Corollary~\ref{cor:imnu}, we know that $\upsilon(c) \in F$.

\begin{theorem}
\label{theo:trdeval}
For any unramified element $c \in K$, the following diagram is 
commutative:
$$\xymatrix @C=5em {
\Aplus \ar[r]^-{\ev_c} \ar[d]_-{\Trd} & \End_F(K) \ar[d]^-{\tr} \\
\Zplus \ar[r]^-{Z(X) \mapsto \upsilon(c)} & F}$$
where $\tr$ denotes the usual trace map over $\End_F(K)$.
\end{theorem}

The rest of this subsection is devoted to the proof of
Theorem \ref{theo:trdeval}.
Let $I$ be the ideal of $\Zplus$ generated by $Z(X) - \upsilon(c)$.
It follows from Proposition~\ref{prop:kerevc}
that $\ev_c$ induces an isomorphism of rings:
$$\alpha : \Aplus/I\Aplus \stackrel\sim\longrightarrow \End_F(K).$$
On the other hand, noticing that $\Trd$ acts on the central element
$Z(X)-\upsilon(c)$ by multiplication by $[K:F]$, we find that $\Trd$ 
induces a map 
$$\beta : \Aplus/I\Aplus \longrightarrow \Zplus/I \simeq F$$
the identification between $\Zplus/I$ and $F$ being induced by the
map $Z(X) \mapsto \upsilon(c)$. After these observations, the 
theorem reduces to proving that $\beta = \tr \circ \alpha$. 
For this, we rely on the following classical characterization of 
the trace map (which we reprove for completeness).

\begin{lemma}
Let $\varphi : \End_F(K) \to F$ be a $F$-linear map such that
$\varphi(uv) = \varphi(vu)$ for all $u, v \in \End_F(K)$. Then,
there exists $\lambda \in F$ such that $\varphi = \lambda \cdot
\tr$.
\end{lemma}

\begin{proof}
Fixing a basis, we can assume that the domain of $\varphi$ is
$M_s(F)$. For $1 \leq i,j \leq s$, let $E_{ij}$ be the matrix
whose unique nonzero entry is located at position $(i,j)$ and
is equal to $1$. If $i \neq j$, we have the relations
$E_{ij} E_{jj} = E_{ij}$ and $E_{jj} E_{ij} = 0$. Therefore
applying $\varphi$, we get $\varphi(E_{ij}) = 0$.
On the other hand, the fact that the matrices $E_{ii}$ and
$E_{jj}$ are conjugated implies that $\varphi(E_{ii}) = 
\varphi(E_{jj})$. By $F$-linearity, we deduce that $\varphi$
must be a scalar multiple of the trace map.
\end{proof}

Applying the previous lemma with $\varphi = \beta \circ \alpha^{-1}$,
we conclude that there exists $\lambda \in F$ with the property that:
\begin{equation}
\label{eq:trdeval}
\Trd(f) \equiv \lambda \cdot \tr \big(\ev_c(f)\big) \pmod I
\end{equation}
for all Ore polynomial $f \in \Aplus$. We have to prove that
$\lambda = 1$.
When $\theta \neq \id_K$, we pick an element $a \in K$ whose trace
over $F$ does not vanish.
Substituting $f = a$ in Eq.~\eqref{eq:trdeval} and noticing 
that $\ev_c(a) : K \to K$ is the multiplication by $a$, we find
$\tr_{K/F}(a) = \lambda \cdot \tr_{K/F}(a)$. Hence $\lambda = 1$
as wanted.

We now consider the case where $\theta = \id_K$. The field $F$ is then 
the subfield of constants of $\delta$ and the polynomial $Z(X)$ is now 
given by Eq.~\eqref{eq:ZX2}.
In accordance with Proposition~\ref{prop:trddiff}, we set
$\tau(a) = \sum_{i=0}^r z_i  \delta^{p^i-1}(a)$
for $a \in K$.
This defines a function $\tau$ which can be considered as a differential 
analogue of the trace map\footnote{Note that, in the differential 
setting, the extension $K/F$ is purely inseparable, so the usual trace 
map $\tr_{K/F}$ vanishes.}. The following lemma summarizes the main
properties of $\tau$.

\begin{lemma}
\label{lem:kerimtau}
The function $\tau$ is $F$-linear and maps surjectively $K$ onto~$F$.
Moreover $\ker\tau = \im\delta$.
\end{lemma}

\begin{proof}
The fact that $\tau$ is $F$-linear is a straightforward verification.
Similarly, we check that the composite $\delta \circ \tau$ vanishes.
Therefore $\tau$ takes its values in $F$.
Thanks to linearity, the surjectivity of $\tau$ will follow if we 
prove that $\tau$ is nonzero. But the vanishing of $\tau$ would mean
that $\delta$ is annihilated by a polynomial of degree $p^r{-}1$, 
which contradicts the definition of $Z(X)$.
It remains to prove that $\ker\tau = \im\delta$. For this, we observe
that $\tau \circ \delta = 0$ and hence that $\ker\tau \subset \im
\delta$. The equality follows by comparing dimensions (over $F$).
\end{proof}

We are now ready to complete the proof of Theorem~\ref{theo:trdeval}.
We fix an element $a \in K$ with $\tau(a) = 1$.
Substituting $f = a (X{-}c)^{p^r -1}$ in Eq.~\eqref{eq:trdeval}, we 
find:
\begin{equation}
\label{eq:trdevaldiff}
\Trd\big(a{\cdot}(X{-}c)^{p^r -1}\big) = \lambda \cdot 
\tr\big(a \delta^{p^r - 1}\big).
\end{equation}
Noticing that $a{\cdot}(X{-}c)^{p^r -1}$ is the sum of $a X^{p^r-1}$
and of terms of smaller degrees, we deduce from 
Proposition~\ref{prop:trddiff} that the reduced trace of
$a{\cdot}(X{-}c)^{p^r -1}$ is $\tau(a) = 1$. On the other hand, 
we can write $\delta^{p^r - 1} = \tau - \sum_{i=0}^{r-1} z_i 
\delta^{p^i-1}$ and get:
\begin{align*}
\lambda \cdot \tr\big(a \delta^{p^r - 1}\big)
 & = \lambda \cdot \tr(a\tau) - 
     \sum_{i=0}^{r-1} \lambda \cdot \tr(a z_i \delta^{p^i - 1}) \\
 & = \lambda \cdot \tr(a\tau) -
     \sum_{i=0}^{r-1} \Trd\big(a z_i (X{-}c)^{p^i - 1}\big)
   = \lambda \cdot \tr(a\tau).
\end{align*}
In order to compute the trace of $a\tau$, we consider a $F$-basis
$(b_1, \ldots, b_{p^r-1})$ of $\ker\tau$. The family 
$(b_1, \ldots, b_{p^r-1}, a)$ is then a $F$-basis of $K$ in which
the matrix of $a\tau$ has all entries equal to $0$ except the one in
the bottom right corner which is $1$. Hence the trace of $a \tau$
is $1$.
Plugging the values we found in both sides of 
Eq.~\eqref{eq:trdevaldiff}, we end up with $\lambda = 1$, which 
concludes the proof.

\section{Residues of Ore rational functions}
\label{sec:residue}

Another important ingredient in the task of determining the duals
of linearized Reed-Solomon codes is the extension of the notion of
residues to Ore polynomials. This extension was already achieved
in~\cite{caruso} in the case where the derivation $\delta$ is zero.
In this section, we address the complementary case $\theta = \id_K$.
By Hilbert's reduction (see Propositions~\ref{prop:hilbert},
\ref{prop:theta-derivations} and the
subsequent discussion), this will cover all cases.

Throughout this section, we then assume that $\theta = \id_K$. In other 
words, we work with a field $K$ equipped with a derivation $\delta : K 
\to K$. We denote by $F$ the subfield of constants and we assume that 
$K/F$ is a finite extension.
It follows from the fact that $F$ contains all $p$-th powers that $K/F$ 
is purely inseparable and hence has degree $p^r$ for some integer~$r$.
As in the previous sections, we denote by $Z(X)$ the minimal polynomial 
of $\delta$ over $K$; it takes the form:
$$Z(X) = X^{p^r} + z_{r-1} X^{p^{r-1}} + \cdots + 
z_1 X^p + z_0 X$$
with all the coefficients $z_i$ in $F$. For convenience, we also
define $z_r = 1$. 
We set $\Aplus = K[X;\id_K,\delta]$ and define the commutative
subrings $\Cplus = K[Z(X)]$ and $\Zplus = F[Z(X)]$. The latter is
the centre of $\Aplus$.

\subsection{Preliminaries}

\subsubsection{Differential trace and differential norm}

In \S \ref{ssec:evalOre} and \S \ref{ssec:trdeval}, we have introduced 
two functions $\upsilon : K \to F$ and $\tau : K \to F$ which, roughly 
speaking, play the role of the norm map and the trace map respectively 
in the differential setting.
For future use, it will be convenient to extend those two functions 
to $\Cplus$. For this, we first extend $\delta$ to a derivation of
$\Cplus$ by letting it act coefficientwise, \emph{i.e.}
$$\delta\Bigg(\sum_{i=0}^d a_i Z(X)^i\Bigg) =
\sum_{i=0}^d \delta(a_i) Z(X)^i.$$
One checks that $\delta$ continues to satisfy the Leibniz rule and
that it takes its values of $\Zplus$.

\begin{definition}
For $C \in \Cplus$, we set:
$$\tau(C) 
 = \sum_{i=0}^r z_i\delta^{p^i-1}(C) 
\quad \text{and} \quad
\upsilon(C) = 
 \sum_{i=0}^r \sum_{j=0}^i \left(z_i\delta^{p^j-1}(C) \right)^{p^{i-j}}.$$
\end{definition}

The above definition gives rise to two functions $\tau : \Cplus 
\to \Zplus$ and $\upsilon : \Cplus \to \Zplus$ that we call the 
\emph{differential trace map} and the \emph{differential norm map} 
respectively. We observe that $\tau$ is $\Zplus$-linear and that 
$\upsilon$ is additive. In addition, the next lemma shows that the 
differential trace is somehow the derivative of the differential norm 
as in the classical setting.

\begin{lemma}
\label{lem:diffupsilon}
For $\varepsilon \in \Zplus$ and $C \in \Cplus$, we have 
$\upsilon( \varepsilon C) \equiv \varepsilon \: \tau(C) 
\pmod{\varepsilon^2}$.
\end{lemma}

\begin{proof}
By definition
$$\upsilon(\varepsilon C)
= \sum_{i=0}^r \sum_{j=0}^i \left(z_i\delta^{p^j-1}(\varepsilon C) \right)^{p^{i-j}}
= \sum_{i=0}^r \sum_{j=0}^i \left(z_i\varepsilon \delta^{p^j-1}(C) \right)^{p^{i-j}}$$
the second equality being correct since $\varepsilon$ is in $\Zplus$ by
assumption.
We observe that only the terms with $i = j$ survive modulo 
$\varepsilon^2$, which gives the lemma.
\end{proof}

Besides, the next proposition shows that the differential norm is 
closely related to the computations in the noncommutative ring 
$\Aplus$.

\begin{proposition}
\label{prop:ZXC}
For $C \in \Cplus$, the identity
$Z(X + C) = Z(X) + \upsilon(C)$
holds in~$\Aplus$.
\end{proposition}

\begin{proof}
This is a direct consequence of \cite[No~12]{jacobson} (see in
particular Eq.~(30) and Eq.~(35)).
\end{proof}

\subsubsection{The fraction field of $\Aplus$}

In the commutative setting, residues have poor interest if we are 
restricting ourselves to polynomials and do not move to the field
of rational functions. In the Ore setting, the same is true.
However, since $\Aplus$ is a noncommutative ring, defining its
field of fractions is not as easy as usual. This can however be 
achieved (see for instance \cite[\S 0.10]{cohn}): 
using Ore condition, one proves
that there exists a unique skew field $\mathcal A$
containing $\Aplus$ for which the following universal property holds:
for any noncommutative ring $\mathfrak{A}$ and any homomorphism of
rings $\phi : \Aplus \rightarrow \mathfrak{A}$ such that $\phi(x)$ 
is invertible for all $x\in \Aplus$, $x \neq 0$, there exists a unique 
morphism of rings
$\psi : \mathcal{A} \rightarrow \mathfrak{A}$ making the 
following diagram commutative:
\begin{equation}
\label{diag:fracA}
\xymatrix @C=5em {
\Aplus \ar[d] \ar[r]^-{\phi} & \mathfrak A \\
\mathcal A \ar[ru]_-{\psi} }
\end{equation}

\noindent
Such a ring $\mathcal A$ is called the fraction field of $\Aplus$. 
In our particular setting, it turns out that one has a rather simple
description of $\mathcal A$.

\begin{proposition}
\label{prop:fractionfield}
The fraction field of $\Aplus$ is 
$\calA = \Frac(\Zplus) \otimes_{\Zplus} \Aplus$.
\end{proposition}

\begin{proof}
We first claim that any Ore polynomial $P \in \Aplus$ has a nonzero
left and a right multiple in $\Zplus$. 
Indeed, observe that the quotient $\Aplus/P\Aplus$ is a finite 
dimensional vector space over $F$. Therefore there exists a 
nontrivial relation of linear dependance of the form:
$$\sum_{i=0}^n a_iZ(X)^i \in P\Aplus \quad (a_i \in F).$$ 
Thus, there
exists $Q \in \Aplus$ with the property that $PQ \in \Zplus$.
Besides, since $PQ$ is central, we deduce that $QPQ = PQQ$ and,
simplifying by $Q$ on the right, we find that $P$ and $Q$ commute.

We are now ready to prove that $\Frac(\Zplus) \otimes_{\Zplus} \Aplus$
is a skew field.
Indeed, reducing to the same denominator, we remark that any
nonzero
element of $\Frac(\Zplus) \otimes_{\Zplus} \Aplus$ can be written as
$D^{-1} \otimes P$ where $D \in \Zplus$, $P \in \Aplus$ and both of
them do not vanish.
By the first part of the proof, there exists $Q \in \Aplus$
such that $PQ = QP \in \Zplus$. Letting $N = PQ$, we check that
$N^{-1} \otimes QD$ is a multiplicative inverse of $D^{-1} \otimes P$.

We consider a noncommutative ring $\mathfrak A$ together with a
ring homomorphism $\varphi : \Aplus \to \mathfrak A$ such that
$\varphi(P)$ is invertible for all $P \in \Aplus$, $P \neq 0$. If
$\psi : \mathcal A \to \mathfrak A$ is an extension
of $\varphi$, it must satisfy:
\begin{equation}
\label{eq:psifracA}
\psi\big(D^{-1} \otimes P\big) = \varphi(D)^{-1} \cdot \varphi(P).
\end{equation}
This proves that, if such an extension exists, it is unique. On the
other hand, using that $\Zplus$ is central in $\Aplus$, one checks that
the formula \eqref{eq:psifracA} determines a well-defined ring
homomorphism $\mathcal A \to \mathfrak A$ making the diagram
\eqref{diag:fracA} commutative.
\end{proof}

\subsection{Taylor expansions}

The main ingredient of the theory of differential residues is a
notion of Taylor expansion for elements of $\Aplus$ extending
the one we are familiar with in the commutative case.

\subsubsection{Existence of Taylor expansions}

We consider an element $z \in F$ and set $N = Z(X) - z \in \Zplus$.
The usual Taylor expansion yields an isomorphism of $K$-algebras:
$$\begin{array}{rcl}
\displaystyle
\varprojlim_{m > 0} \Cplus/N^m\Cplus & 
\stackrel\sim\longrightarrow & K\croT \vspace{-1ex}\\
f(Z(X)) & \mapsto & f(z) + f'(z) T + \cdots +
\frac{f^{(n)}(z)}{n!} T^n + \cdots
\end{array}$$
which is uniquely determined by the fact that it maps $N$ to $T$ and
it induces the identify after quotienting out by $N$ on the left
and by $T$ on the right. The next theorem tells that this
isomorphism extends to $\Aplus$.

\begin{theorem}
\label{theo:taylorexist}
With the above notations, there exists an isomorphism of $K$-algebras:
$$\varprojlim_{m > 0} \Aplus/N^m\Aplus 
\stackrel\sim\longrightarrow (\Aplus/N\Aplus)\croT$$
sending $N$ to $T$ and inducing the identity when we quotient out by 
$N$ on the left and by $T$ on the right.
\end{theorem}

\begin{proof}
Throughout the proof, we fix an element $a \in K$ such as $\tau(a) = 
1$; such an element exists thanks to Lemma~\ref{lem:kerimtau}.
We are going to construct by induction a sequence $(\zeta_m)_{m > 0}$
of elements of $\Zplus$ such that $\zeta_1 = 0$ and, for $m > 0$,
\begin{itemize}
\item $\zeta_{m+1} \equiv \zeta_m \pmod {N^m}$,
\item $N(X + a \zeta_m) \in N^m \Aplus$.
\end{itemize}
We suppose that the sequence has been built until the index $m$.
The next term $\zeta_{m+1}$ is of the form $\zeta_m + N^m P$ for 
some polynomial $P \in \Zplus$. Besides it satisfies our requirements 
if and only if $N(X + a \zeta_{m+1}) \in N^{m+1} \Aplus$.
Relying on Lemma~\ref{lem:diffupsilon} and Proposition~\ref{prop:ZXC},
we carry out the following computation:
\begin{align*}
N(X + a \zeta_{m+1}) 
& = N(X) + \upsilon(a \zeta_{m+1})
  = N(X) + \upsilon(a \zeta_m) + \upsilon(a N^m P) \\
& \equiv N(X + a \zeta_m) + N^m P \pmod{N^{m+1}}
\end{align*}
given that $\tau(a) = 1$ and $N^m P \in \Zplus$.
Since $N(X + a \zeta_m)$ is divisible by $N^m$
thanks to our induction hypothesis, one can choose $P$ in order to
ensure that $N(X + a\zeta_{m+1}) \equiv 0 \pmod{N^{m+1}}$. This completes
the construction of $\zeta_{m+1}$.

We now set:
$$\zeta = (\zeta_m)_{m > 0} \in \varprojlim_{m > 0} \Zplus/N^m\Zplus.$$
Passing to the limit, we find $N(X + a\zeta) = 0$.
This property allows us to define a morphism of $K$-algebras:
$$\iota : \Aplus/N\Aplus 
\longrightarrow \varprojlim_{m>0} \Aplus/N^m\Aplus, \quad X \mapsto X+a\zeta.$$
Furthermore, as $C \equiv 0 \pmod N$, $\iota$ reduces to the identity 
map modulo $N$. By sending $T$ to $N$, we can extend $\iota$ to a 
second morphism :
$$\rho : (\Aplus/N\Aplus)\croT 
\longrightarrow \varprojlim_{m>0} \Aplus/N^m\Aplus.$$
This morphism reduces to the identity when we quotient out by $T$ on 
the left and by $N$ on the right. It is moreover bijective since its
domain and codomain are both
separated and complete (for the $T$-adic and $N$-adic topology 
respectively). Its inverse then satisfies all the requirements of the
theorem.
\end{proof}

\subsubsection{Unicity of Taylor expansions}

Unfortunately, unlike the commutative case, an isomorphism satisfying
the conditions of Theorem~\ref{theo:taylorexist} is not unique in general.
For this reason, it is convenient to introduce the following definition.

\begin{definition}
Keeping the above notations, an isomorphism
$$\varprojlim_{m > 0} \Aplus/N^m\Aplus \longrightarrow (\Aplus/N\Aplus)\croT$$
is called \emph{$z$-admissible} (or simply \emph{admissible} if there
is no risk of confusion) if it maps $N$ to $T$ and it induces the
identity after quotienting out by $N$ on the left and by $T$ on the
right.
\end{definition}

\begin{proposition}
\label{prop:taylorunique}
Let
$$\tau_1, \tau_2 :
\varprojlim_{m > 0} \Aplus/N^m\Aplus 
\stackrel\sim\longrightarrow (\Aplus/N\Aplus)\croT$$
be two admissible isomorphisms. 
Then there exists $V \in (\Aplus/N\Aplus)\croT$ with $V 
\equiv 1 \pmod T$ such that $\tau_1(f) = V^{-1}\tau_2(f) V$ for all
$f \in \varprojlim_{m > 0} \Aplus/N^m\Aplus$.
\end{proposition}

\begin{proof}
For simplicity, we write $R = \Aplus/N\Aplus$.
We first claim that $R$ is a simple central algebra over $\Zplus/
N\Zplus \simeq F$. Indeed, it is central because the formation of
centres commutes with the tensor product. In order to prove that
it is simple, let $I$ be a nonzero two-sided ideal of $R$. Since 
it is in particular a right ideal, there exists a monic divisor $P$ 
of $N$ 
such that $I = P \Aplus/N\Aplus$. Observe that the commutator $PX - XP$ 
lies in $I$ and has degree strictly less that $\deg P$. Hence it has to 
vanish, meaning that $PX = XP$ in $R$. Similarly, we prove that $Pa = 
aP$ for all $a \in K$. Therefore $P$ is central in $R$, which shows
that
$P \in \Zplus/N\Zplus$. Since the latter is a field, we deduce that
$P$ is invertible in $R$ and finally that $I = R$. Hence $R$ is 
simple.

Define $\tau = \tau_1 \circ \tau_2^{-1}$; it is an automorphism
of $R\croT$ which takes $T$ to itself and is congruent 
to the identity modulo $T$. Since $R$ is simple central, it follows
from \cite[Theorem~9.1]{BHKV} (applied with $\varphi = \tau_{|R}$) 
that there exists an invertible
element $c \in R\croT$ such that $\tau(x) = c^{-1} x c$ for
all $x \in R$. In fact, the latter equality holds more generally 
for any $x \in R\croT$ given that $\tau(T) = T$.
Finally, the fact that $\tau$ is congruent to the identity modulo 
$T$ indicates that $c_0 = c \text{ mod } T$ must be central in $R$. 
The proposition then holds for $V = c_0^{-1} c$.
\end{proof}



\subsection{Construction of skew residues}

We recall that we have constructed earlier the fraction field of 
$\Aplus$ (see Proposition~\ref{prop:fractionfield}); in what follows, 
we will denote it by $\calA$. Similarly, we set $\calC = \Frac(\Cplus)$ 
and $\calZ = \Frac(\Zplus)$.
Besides, as before, we consider an element $z \in F$ and set 
$N = Z(X)-z \in \Zplus$. We choose an admissible isomorphism
$\tau_z$ and consider the compositum:
$$\TS_z : \Aplus \longrightarrow 
\varprojlim_{m > 0} \Aplus/N^m\Aplus 
\stackrel{\tau_z}\longrightarrow (\Aplus/N\Aplus)\croT$$
where the first map is induced by the canonical projections
$\Aplus \to \Aplus/N^m\Aplus$.

\begin{lemma}
\label{lem:extendTSa}
For any $f \in \calZ$, the series $\TS_z(f)$ is invertible in
$(\Aplus/N\Aplus)\parT$.
\end{lemma}

\begin{proof}
It is enough to prove the lemma when $f$ is monic and irreducible
in~$\Zplus$. It is clear when $f = N$ because $\TS_a$ maps $N$ to
$T$, which is by definition invertible in $(\Aplus/N\Aplus)\parT$.
On the other hand, if $f$ is different from $N$, it must be coprime
with $N$ by irreducibility. It is then invertible in each quotient
$\Zplus/N^m\Zplus$ and thus it is also a unit
in each $\Aplus/N^m\Aplus$. Passing to the limit, we find that $f$
is invertible in $\varprojlim_{m > 0} \Aplus/
N^m\Aplus$; it is then also in $(\Aplus/N\Aplus)\croT$
given that $\tau_a$ is an isomorphism.
\end{proof}

Combining Proposition~\ref{prop:fractionfield} and 
Lemma~\ref{lem:extendTSa}, we find that $\TS_z$ uniquely extends
to a ring homomorphism $\calA \to (\Aplus/N\Aplus)\parT$
that, in a slight abuse of notations, we continue to denote by 
$\TS_z$.

It turns out that the previous construction extends to elements 
lying in extensions of $F$. Precisely, let $\Fs$ denote a fixed 
separable closure of $F$ and set $\Ks = \Fs \otimes_F K$.
Since $K/F$ is purely inseparable, it is linearly disjoint from $\Fs$, 
implying that $\Ks$ is a field. Moreover, the derivation $\delta : K 
\to K$ extends uniquely by $\Fs$-linearity to a derivation of $\Ks$ 
whose field of constant is $\Fs$ and minimal polynomial is still 
$Z(X)$. In what follows, we continue to call $\delta$ this extension. 
We define $\Asplus = \Ks[X; \delta]$ and $\calAs = \Frac(\Asplus)$.
Applying what we have done previously with $K$ replaced by $\Ks$ (and 
$F$ replaced by $\Fs$ accordingly), we end up with a ring homomorphism:
$$\TS_z : \calA \to \big(\Asplus/(Z(X){-}a)\Asplus\big)\parT$$
for any $z \in \Fs$.
The series $\TS_z(f)$ is called the \emph{Taylor expansion} of $f$ 
around $a$.

We insist on the fact that it does depend on a choice of 
the admissible isomorphism $\tau_z$.
However, from Proposition~\ref{prop:taylorunique}, we derive that
two different choices of $\tau_z$ lead to two mappings $\TS_z$
which are conjugated by an element congruent to $1$ modulo~$T$.
In particular, the two following quantities are defined without
ambiguity:
\begin{itemize}
\item the \emph{order of vanishing} of $f$ at $z$, 
denoted by $\ord_z(f)$, defined as the $T$-adic valuation of $\TS_z(f)$,
\item the \emph{principal part} of $f$ at $z$, denoted by 
$\mathcal{P}_z(f)$, defined as the coefficient of $T^{\ord_z(f)}$ in the 
series $\TS_z(f)$.
\end{itemize}

We are now ready to define skew residues.

\begin{definition}
Given $f \in \calA$ and $z \in \Fs$, the \emph{skew residue} 
of $f$ at $z$, denoted by $\sres_z(f)$, is the coefficient of $T^{-1}$ 
in the series $\TS_z(f)$.
\end{definition}

Again, we insist on the fact that skew residues do depend on the choice 
of an admissible isomorphism. However, they are defined without 
ambiguity when the Ore function $f$ has at most a \emph{simple pole}
at the point $z$ we are looking at, \emph{i.e.} if $\ord_z(f) \geq -1$.
We shall see in \S\ref{ssec:trdsres} below that some quantities related to $\sres_a(f)$
are also well-defined in full generality.

\subsection{Reduced traces of skew residues}
\label{ssec:trdsres}

We recall that we have introduced in \S \ref{sec:trd} the reduced 
trace map $\Trd : \Aplus \to \Zplus$ and that we have given an
explicit formula for it in Proposition~\ref{prop:trddiff}.
Using Proposition~\ref{prop:fractionfield}, we extend by 
$\calZ$-linearity the map $\Trd$ to a mapping $\calA \to \calZ$ 
(we recall that $\calZ = \Frac(\Zplus)$) and continue to call $\Trd$ 
this extension.

\begin{proposition}
\label{prop:commTrdsres}
For all $f \in \calA$ and all $z \in \Fs$, we have:
$$\Trd\big(\sres_z(f)\big) = \res_z\big(\Trd(f)\:dZ(X)\big)$$
where $\res_z(\omega)$ denotes the residue at $z$ of the 
differential form $\omega$.
\end{proposition}

\begin{proof}
Write $N = Z(X) - z$. The proposition will follow if you prove the
commutativity of the following diagram:
$$\xymatrix @C=1.5ex {
  \calA \ar[rrrr]^-{\TS_z} \ar[d]_-{\Trd} 
&&&& \big(\Asplus/N\Asplus\big)\parT \ar[d]^-{\Trd} \\
  \calZ \ar[rrrr]^-{\TS_z}
&&&& \big(\Zsplus/N\Zsplus\big)\parT \ar@{=}[r] & \Fs\parT. }$$
By $\calZ$-linearity, it is enough to consider the case where
$f \in \Aplus$. We equip $\Aplus$ (resp. $\Asplus/N\Asplus$) 
with its canonical basis  $(1, X, \ldots, X^{p^r-1})$ over $\Cplus$
(resp. $\Csplus/N\Csplus$).
Let $M$ be the matrix of the map $x \mapsto xf$ acting on $\Aplus$
and similarly, let $N$ be the matrix of the map $x \mapsto x\cdot
\TS_z(f)$ acting on $\big(\Asplus/N\Asplus\big)\croT$.
By definition $\Trd(f)$ is the trace of $M$ while $\Trd\circ
\TS_z(f)$ is the trace of $N$. On the other hand, from the fact
that $\TS_z$ is a ring homomorphism, we deduce that the matrices
$\TS_z(M)$ and $N$ and conjugated so, they have the same trace.
Hence $\Trd\circ\TS_z(f) = \TS_z\circ\Trd(f)$ and we have proved
our claim.
\end{proof}

\begin{remark}
Proposition~\ref{prop:commTrdsres} can be refined as follows.
Let $\sigma_0 : \Aplus \to \Cplus$ be the $\Cplus$-linear form
defined by $\sigma_0(X^i) = 0$ for $0 \leq i < p^r - 1$ and
$\sigma_0(X^{p^r-1}) = 1$. Extending scalars, we find that 
$\sigma_0$ induces mappings $\calA \to \calC$ and 
$(\Asplus/N\Asplus)\parT \to (\Csplus/N\Csplus)\parT$. With these
notations, one can prove that:
$$\sigma_0\big(\sres_z(f)\big) = \res_z\big(\sigma_0(f)\:dZ(X)\big)$$
This latter statement gives back Proposition~\ref{prop:commTrdsres}
after applying $\tau$ on both sides; it then indeed appears as a
refinement of the proposition.
\end{remark}

Proposition~\ref{prop:commTrdsres} admits several interesting 
corollaries. For example, it shows that the reduced trace of 
$\sres_z(f)$ is canonical in the sense that it does not depend on a 
choice of an admissible isomorphism $\tau_z$.
Besides, one has a noncommutative analogue of the residue formula
given by the next theorem.

\begin{theorem}[Residue formula]
\label{theo:residueformula}
Let $f = \frac P D \in \calA$ with $P \in \Aplus$ and $D \in \Zplus$,
$D \neq 0$. If $\deg P \leq \deg D - 2$, we have:
$$\sum_{z \in F^s} \Trd\big(\sres_z(f)\big) = 0.$$
\end{theorem}

\begin{proof}
Write $g = \Trd(f) = \frac{\Trd(P)} D \in \calZ$ and define the
differential form $\omega = g {\cdot} dZ(X)$.
Applying Proposition~\ref{prop:commTrdsres} to each summand, we obtain:
\begin{equation}
\label{eq:resformula}
\sum_{z \in F^s} \Trd\big(\sres_z(f)\big) =
\sum_{z \in F^s} \res_z(\omega).
\end{equation}
By the condition on the degrees, the differential form $\omega$
has no pole at infinity. It may have poles at
inseparable points but the corresponding residues all vanish.
From the classical residue formula, we then deduce that the
right hand side of Eq.~\eqref{eq:resformula} vanishes. The left
hand side then vanishes as well, establishing the theorem.
\end{proof}

\section{Duality over Ore polynomial rings}
\label{sec:duality}

In this section, we discuss duality over Ore rings and its relation
with evaluation morphisms and residue maps.

We keep the framework of the previous sections:
we consider a field $K$ equipped with
a ring homomorphism $\theta : K \rightarrow K$ and a
$\theta$-derivation $\delta : K \rightarrow K$.
We denote by $F$ the subfield of $K$ consisting of elements $a \in K$
such that $\theta(a) = a$ and $\delta(a) = 0$ and we assume as always
that $K/F$ is a finite extension. We set $\Aplus = K[X; \theta,
\delta]$.

We recall that the centre $\Zplus$ of $\Aplus$ is the ring of
univariate polynomials over $F$ in a distinguished element $Z(X)$.
When $\theta = \id_K$, this special central polynomial $Z(X)$ takes
the form
\begin{equation}
\label{eq:ZX3}
Z(X) = X^{p^r} + z_{r-1} X^{p^{r-1}} + \cdots + 
z_1 X^p + z_0 X \quad (z_i \in F)
\end{equation}
(see Eq.~\eqref{eq:ZX}). On the contrary, when $\theta \neq \id_K$,
we have $Z(X) = (X{+}a)^s$ where $s = [K:F]$ is the order of $\theta$
and $a$ is some element of $K$. If $\theta = \id_K$, we define $s = 
p^r$. In both cases, we then have $s = \deg Z = [K:F]$.

\subsection{Two perfect pairings}
\label{ssec:pairing}

To begin with, we recall the definition of duality in the
context of the sum-rank metric. Our presentation differs slightly
from the most usual one in the sense we do not work with matrices
and transposes but instead with pairings and adjoints.

When $\theta = \id_K$, we define $\tau : K \to F$ as in \S
\ref{ssec:trdeval} by the formula:
$$\tau(f) = \sum_{i=0}^r z_i \delta^{p^i-1}(f).$$
On the contrary, when $\theta \neq \id_K$, we let $\tau$ denote
the trace map of $K/F$.

\begin{proposition}
The pairing $K \times K \to F$, $(f,g) \mapsto \langle f,g\rangle_K = 
\tau(fg)$ is $F$-bilinear and nondegenerate.
\end{proposition}

\begin{proof}
Bilinearity follows from the linearity of $\tau$.
When $\theta = \id_K$, nondegeneracy follows from the fact that $\tau : 
K \to F$ is surjective (see Lemma~\ref{lem:kerimtau}).

When $\theta \neq \id_K$, we claim that $F$ is the subfield of $K$ 
fixed by $\theta$. Indeed, by Proposition~\ref{prop:theta-derivations}, 
we know that $\delta$ is a scalar multiple of $\delta_0 = \theta - 
\id_K$. Therefore $\delta$ vanishes on the subfield fixed by $\theta$ 
and our claim is proved.
It follows that the extension $K/F$ is separable, implying eventually
that $\tau$ is nondegenerate.
\end{proof}

From now on, we endow $K$ with the bilinear pairing 
$\langle f,g\rangle_K = \tau(fg)$. If $V$ is a 
$F$-linear subspace of $K$, we recall that the \emph{orthogonal}
of $V$, denoted by $V^\perp$, is defined as the set of all $f \in K$ 
such that $\langle f,g\rangle_K = 0$ for all $g \in V$.
Similarly, if $\varphi : K \to K$ is a $F$-linear map, we define
the \emph{adjoint} of $\varphi$, denoted by $\varphi^\star$, as
the endomorphism of $K$ determined by the adjunction rule:
$$\langle \varphi^\star(f), g \rangle_K = \langle f, \varphi(g)\rangle_K$$
which is required to hold true for all $f, g \in K$.
It is important to notice that for all $a \in K$,
the multiplication by $a$ (\emph{i.e.} the mapping $\mu_a : K \to K$, 
$x \mapsto ax$) is self-adjoint (\emph{i.e.} $\mu_a^\star = \mu_a$).

It is well-known that $\ker(\varphi^\star) = \im(\varphi)^\perp$ and 
$\im(\varphi^\star) = \ker(\varphi)^\perp$. From these equalities, it 
follows that $\varphi$ vanishes on some $F$-linear subspace $V$ of $K$ 
if and only if $\varphi^\star$ takes its values in $V^\perp$. 
The adjoint construction then induces a bijection between 
$\Hom_F(K/V,K)$ and $\Hom_F(K,V^\perp)$. Replacing $V$ by $V^\perp$, we 
find that is also induces an isomorphism $\Hom_F(K/V^\perp,K) 
\stackrel\sim\rightarrow \Hom_F(K,V)$.

\begin{lemma}
\label{lem:pairing}
For any $F$-linear subspace $V$ of $K$, the mapping :
$$\begin{array}{rcl}
\Hom_F(K/V^\perp,K) \times \Hom_F(V,K) & \longrightarrow & F \smallskip \\
(\varphi, \psi) & \mapsto & 
\langle \varphi, \psi \rangle = \tr(\varphi^\star \circ \psi)
\end{array}$$
(where $\tr$ denotes the trace map)
is a perfect $F$-bilinear pairing.
\end{lemma}

\begin{proof}
As $\Hom_F(K/V^\perp,K)$ and $\Hom_F(V,K)$ have the same dimension, it is 
enough to establish the following property: if $\psi \in \Hom_F(V,K)$ 
is such that $\tr(\varphi^\star \circ \psi) = 0$ for all $\varphi \in \Hom_F(K/V^\perp,K)$, 
then $\psi = 0$. Given that the adjunction induces an isomorphism 
between $\Hom_F(K/V^\perp,K)$ and $\Hom_F(K, V)$, it is enough to 
prove that $\psi = 0$ if $\tr(\varphi \circ \psi) = 
0$ for all $\varphi \in \Hom_F(K,V)$.
Taking basis and writing $s = [K{:}F]$, $d = \dim_F V$, 
we are reduced to check that if a matrix
$M \in \mathcal{M}_{s,d}$
satisfies $\tr(NM) = 0$ for all $N \in \mathcal{M}_{d,s}$, then
$M$ is zero. This finally follows from the observation that $\tr(NM)$ is 
the $(i,j)$ coefficient of $M$ when $N$ is the matrix with all entries 
equal to $0$ except the one in position $(j,i)$ which is equal to $1$.
\end{proof}

\begin{remark}
An important particular case occurs when $V = K$; in this situation
$K/V^\perp$ is equal to $K$ as well and the bilinear form 
$\langle -, - \rangle$ of Lemma~\ref{lem:pairing} defines a perfect 
pairing over $\End_F(K)$.
\end{remark}

If $C$ is a $F$-linear subspace of $\Hom_F(V,K)$ (resp. of
$\Hom_F(K/V^\perp,K)$), we will denote by $C^\perp$ its orthogonal in
$\Hom_F(K/V^\perp,K)$ (resp. in $\Hom_F(V,K)$). Since our pairing is
nondegenrate, we always have $(C^\perp)^\perp = C$ and the following
equality of dimensions:
\begin{equation}
\label{eq:dimperp}
\dim_F C + \dim_F C^\perp = \dim_F \Hom_F(V,K) =
[K:F] \cdot \dim_F V.
\end{equation}
Besides, it is worth noticing that both $\Hom_F(K/V^\perp,K)$ and
$\Hom_F(V,K)$ are endowed with a natural structure of $K$-linear
vector spaces since $K$ acts on the codomains.
The next lemma ensures that $K$-linearity is preserved under duality.

\begin{lemma}
\label{lem:perpKlin}
If $C$ is a $K$-linear subspace of $\Hom_F(V,K)$, 
then $C^\perp$ is a $K$-linear subspace of $\Hom_F(K/V^\perp,K)$.
Moreover, we have:
$$\dim_K C + \dim_K C^\perp = \dim_F V.$$
\end{lemma}

\begin{proof}
Let $\varphi \in C^\perp$. 
Let $a \in K$ and let $\mu_a : K \to K$ denote the multiplication map
by $a$. Given $\psi \in C$, we compute:
$$\tr\big((\mu_a \circ \varphi)^\star \circ \psi\big)
= \tr\big(\varphi^\star \circ \mu_a \circ \psi\big) = 0$$
the first equality coming from the fact that $\mu_a$ is self-adjoint
while the second one is correct because $\mu_a \circ \psi$ is in $C$ 
given that $C$ is a $K$-linear subspace by assumption. Consequently
$\mu_a \circ \varphi \in C^\perp$ and we have proved that $C^\perp$ is
stable under multiplication by $K$.

Finally, the equality of dimensions follows immediately 
from Eq.~\eqref{eq:dimperp}.
\end{proof}

\subsection{Construction of duality}
\label{ssec:dualOre}

We recall that we have proved in \S \ref{ssec:Ore}
that the pair $(\theta, \delta)$ is either of the form
$(\id_K, \delta)$ with $\delta \neq 0$ or $(\theta, a \delta_0)$
with $a \in K$ and $\delta_0 = \theta{-}\id_K$.
We recall also that we have defined the notion of ramified elements 
in Definition~\ref{def:ramified}. When $\theta = \id_K$, all elements 
are actually unramified whereas there is exactly one ramified element
in the second case, which is $-a$.
Accordingly we define:
$$\begin{array}{r@{\hspace{0.5ex}}l@{\qquad}l}
\Aur 
  & = \Aplus
  & \text{if } \theta = \id_K \smallskip \\
  & = \Aplus\big[\frac 1{X{+}a}] 
    = \Zplus\big[\frac 1{(X{+}a)^s}] \otimes_{\Zplus} \Aplus
  & \text{if } \theta \neq \id_K \text{ and } \delta = a \delta_0.
\end{array}$$
where we recall that $s$ is the degree of the extension $K/F$.

\begin{definition}
For $f \in \Aur$, we define $f^\star$ as follows.
\begin{itemize}
\item If $\theta = \id_K$, we write $f = \sum_{i=0}^n a_i X^i$ and set:
$$f^\star = \sum_{i=0}^n (-1)^i X^i a_i.$$
\item If $\theta \neq \id_K$ and $\delta = a \delta_0$, we write
$f = \sum_{i=v}^n a_i (X{+}a)^i$ and set:
$$f^\star = \sum_{i=v}^n (X{+}a)^{-i} a_i.$$
\end{itemize}
\end{definition}

It is a straightforward calculation in check in both cases that the
three following properties hold true for all $f, g$ in $\Aur$:
(i)~$f^{\star\star} = f$,
(ii)~$(f+g)^{\star} = f^\star + g^\star$ and (iii)~
$(fg)^\star = g^\star f^\star$. In other words, the duality
construction $f \mapsto f^\star$ defines a ring homomorphism
$(\Aur)^\op \to \Aur$ which is an involution. Here $(\Aur)^\op$
denotes the opposite ring of $\Aur$ which is defined by reversing
the direction of the multiplication. By the universal property of
the fraction field, the duality extends to a ring homomorphism
$\calA^\op \to \calA$. Notice that:
$$\begin{array}{r@{\hspace{0.5ex}}l@{\qquad}l}
Z(X)^\star 
 & = -Z(X) & \text{if } \theta = \id_K
  \smallskip \\
 & = Z(X)^{-1} & \text{otherwise.}
\end{array}$$

We now want to compare the above duality with the evaluation maps 
$\ev_c$ we have introduced in \S \ref{ssec:evalOre}. We recall that 
we have already computed the kernel of $\ev_c$ in 
Proposition~\ref{prop:kerevc}; it is the principal ideal generated by 
$Z(X){-}\upsilon(c)$ where $\upsilon(c)$ is given explicitely by 
Eq.~\eqref{eq:upsilon} when $\theta = \id_K$ and is equal to 
$N_{K/F}(c{+}a)$ otherwise.

Given an unramified element $c \in K$, it is convenient to put:
$$\begin{array}{r@{\hspace{0.5ex}}l@{\qquad}l}
c^\vee
 & = -c & \text{if } \theta = \id_K
  \smallskip \\
 & = \frac 1{c+a} - a
 & \text{if } \theta \neq \id_K \text{ and } \delta = a \delta_0.
\end{array}$$

\begin{theorem} \label{theo:evaldual}
Let $c \in K$ be an unramified element and set $N = Z(X) - \upsilon(c)$.
The following diagram is commutative:
\begin{center}
\begin{tikzpicture}
  \matrix (m) [matrix of math nodes,row sep=3em,column sep=4em,minimum width=2em]
  {
     \Aur/N\Aur & \End_F(K) \\
     \Aur/N^\star\Aur & \End_F(K) \\};
  \path[-stealth]
    (m-1-1) edge node [left] {$f \mapsto f^\star$} (m-2-1)
            edge node [above] {$\ev_c$} (m-1-2)
    (m-2-1.east|-m-2-2) edge node [below] {$\ev_{c^\vee}$}
            node [above] {} (m-2-2)
    (m-1-2) edge node [right] {$\varphi \mapsto \varphi^\star$} (m-2-2);
\end{tikzpicture}
\end{center}
where $\varphi^\star$ is the adjoint of $\varphi$ for the pairing
$\left< x,y \right>_K = \tau(xy)$ as in \S \ref{ssec:pairing}.
\end{theorem}

\begin{proof}
By linearity and using the fact that the multiplication by elements of 
$K$ are self-adjoint, it is enough to prove the theorem when $f = X^i$ 
(resp. $f = (X{+}a)^i$) for some $i$ when $\theta = \id_K$ (resp. 
$\theta \neq \id_K$). By the multiplicativity property of adjoints,
this further amounts to checking that $\delta^\star = -\delta$ (resp. 
$\theta^\star = \theta^{-1}$).

We first consider the case where $\theta = \id_K$. Let
$x, y \in K$. From the relation $\tau(\delta(xy)) = 0$ (see
Lemma~\ref{lem:kerimtau}), we derive $\tau(x\delta(y)) = -\tau(\delta(x)y)$,
which also reads
$\left< x, \delta(y) \right>_K = \left< -\delta(x), y \right>_K$.
We can then conclude in this case.
In the case where $\theta \neq \id_K$, we need to show that $\langle 
\theta(x), y \rangle_K = \langle x, \theta^{-1}(y) \rangle_K$ for $x,y 
\in K$. By definition, this reduces to verify that 
$\tau\big(\theta(x){\cdot} y\big) = \tau\big(x{\cdot} \theta^{-1}(y)\big)$ 
which is
obvious because $\tau$, being the trace map $\tr_{K/F}$ in this case,
takes the same value on two conjugated elements.
\end{proof}

We now focus on the comparison between duality and residues.
We recall that residues have been constructed in~\cite{caruso}
in the case of $K[X;\theta,0]$ and in \S \ref{sec:residue} in the 
case of $K[X;\id_K,\delta]$. Using Hilbert twists (see
Propositions~\ref{prop:hilbert} and~\ref{prop:theta-derivations}), 
these two special situations cover all cases.

\begin{theorem} \label{theo:resdual}
Let $z \in \Fs$ and assume that $z \neq 0$ if $\theta \neq \id_K$.
Set $\bar z = -z$ if $\theta = \id_K$ and $\bar z = z^{-1}$ 
otherwise. Then, for all $f \in \calA$ and for any 
$z$-admissible isomorphism, there exists a
$\bar z$-admissible isomorphism such that:
$$\begin{array}{r@{\hspace{0.5ex}}l@{\qquad}l}
\sres_z(f^\star) 
 & = - \sres_{\bar z}(f)^\star
 & \text{if } \theta = \id_K \smallskip \\
 & = - Z(X)^{-2} \cdot \sres_{\bar z}(f)^\star
 & \text{otherwise}
\end{array}$$
(where the skew residues are computed according to the 
corresponding choices of admissible isomorphisms).
\end{theorem}

Before giving the proof of Theorem~\ref{theo:resdual}, we record
the following lemma.

\begin{lemma} \label{changement de variable}
We keep the assumptions of Theorem~\ref{theo:resdual}.
Set $N = Z(X) - z$ and
let $S \in \Aplus/N \Aplus\croT$ be a series with constant term $0$. 
Let :
$$\begin{array}{rcl}
\psi : \quad \Aplus/N\Aplus\parT & 
  \longrightarrow & \Aplus/N\Aplus\parT \smallskip \\
\displaystyle \sum_i a_iT^i & \mapsto & \displaystyle \sum_i a_iS^i.
\end{array}$$
For all $f \in \Aplus/N\Aplus\parT$, we have the formula :
$$\res \left(\psi(f)\frac{\partial S}{\partial T}\right) = \res(f),$$
where $\res$ is the application selecting the coefficient in $T^{-1}$.
\end{lemma}

\begin{proof}
If $f \in \Aplus/N\Aplus\croT$, both sides of the 
formula vanish and the lemma holds.
As $\res$ and $\psi$ are $K$-linear, it is enough to verify the lemma 
when $f = T^i$ for $i<0$. Then, the formula becomes
$\res \big(S^i\frac{\partial S}{\partial T}\big) = \res(T^i)$, which is 
a direct consequence of the classical formula of change of variables
for residues.
\end{proof}

\begin{proof}[Proof of Theorem~\ref{theo:resdual}]
We consider a $z$-admissible isomorphism:
$$\tau_z : 
\varprojlim_{m > 0} \Aplus/N^m\Aplus 
\stackrel\sim\longrightarrow (\Aplus/N\Aplus)\croT.$$
Conjugating it by the duality on both sides, we end up with a second
isomorphism:
$$\tau_z^\star : 
\varprojlim_{m > 0} \Aplus/N^{\star m}\Aplus 
\stackrel\sim\longrightarrow (\Aplus/N^\star\Aplus)\croT.$$
Write $S = -T$ if $\theta = \id_K$ and $S = \frac{-z^2T}{1+zT}$
otherwise.
If $\psi$ is the corresponding morphism of
Lemma~\ref{changement de variable}, an easy computation shows that
$\tau_{\bar z} = \psi \circ \tau_z^\star$ is $\bar z$-admissible.
Theorem~\ref{theo:resdual} now follows from 
Lemma~\ref{changement de variable} after noticing that 
$\frac{\partial S}{\partial T} = -1$ if $\theta = \id_K$
and:
$$\frac{\partial S}{\partial T} = - \frac{z^2}{(1+zT)^2} =
\tau_{\bar z}\big({-}Z(X)^{-2}\big)$$
otherwise.
\end{proof}

\section{Duals of linearized Reed-Solomon codes}
\label{sec:LG}

After all the preparations achieved in the previous sections, we are 
finally ready to give an explicit construction of the duals of the 
Martínez-Peñas' linearized Reed-Solomon codes.

We come back to the setting of \S \ref{ssec:LRS}: in addition of
$K$, $\theta$ and $\delta$, we consider a positive 
integer $m$, a tuple $\underline{c} = (c_1, \dots, c_m)$ 
of unramified elements of $K$ and another tuple $\underline{V} = 
(V_1, \dots, V_m)$ of $F$-linear subspaces of $K$. For each index $i$, 
we set $z_i = \upsilon(c_i)$. We always assume that the $z_i$'s are 
pairwise distinct.
We further define $N_i = Z(X) - z_i \in \Zplus$ and $N = \prod_{i=1}^m 
N_i \in \Zplus$. For simplicity, we also write:
\begin{align*}
\Hom_F(\underline V,K) 
& = \Hom_F(V_1,K) \times \cdots \times \Hom_F(V_m,K), \\
\Hom_F(K/\underline V,K) 
& = \Hom_F(K/V_1,K) \times \cdots \times \Hom_F(K/V_m,K), \\
\Hom_F(K/\underline V^\perp,K) 
& = \Hom_F(K/V_1^\perp,K) \times \cdots \times \Hom_F(K/V_m^\perp,K).
\end{align*}
It follows from Lemma~\ref{lem:pairing} that the formula:
$$\big\langle (\varphi_1, \ldots, \varphi_m),\, 
(\psi_1, \ldots, \psi_m)\big\rangle = 
\sum_{i=1}^m \tr(\varphi_i^\star \circ \psi_i)$$
defines a perfect $F$-bilinear pairing between the spaces 
$\Hom_F(K/\underline V^\perp,K)$ and $\Hom_F(\underline V,K)$.
For any $K$-linear code in $\Hom_F(\underline V,K)$, we let
$C^\perp$ denote its orthogonal in $\Hom_F(K/\underline V^\perp,K)$.
From Lemma~\ref{lem:perpKlin}, we derive that $C^\perp$ is
a $K$-linear code of same length and complementary dimension.

Our aim is to give an alternative description of the code $\LRS(k, 
\underline c, \underline V)^\perp$ which sits inside 
$\Hom_F(K/\underline V^\perp,K)$.

\subsection{Linearized Goppa codes}

In this subsection, we introduce a new family of codes for the
sum-rank metric constructing by taking residues, that we call
linearized Goppa codes. We then prove that they those codes are
isomorphic to some linearized Reed-Solomon codes.
We keep the notations $\Aplus$, $\Zplus$, $\calA$, $\calZ$,
\emph{etc.} of the previous sections.
Let $D$ be the monic polynomial attached to the $c_i$'s and $V_i$'s
by the result of Theorem~\ref{theo:multipleeval}.(2).

\begin{lemma}
\label{lem:simplepole}
Let $f \in \Aplus D^{-1}$. Then, for all $i \in \{1, \ldots, m\}$,
$f$ has at most a simple 
pole at $z_i$ and $\ev_{c_i}(\sres_{z_i}(f))$ vanishes on $V_i$.
\end{lemma}

\begin{proof}
We write $f = gD^{-1}$ for some $g \in \Aplus$.
By theorem \ref{theo:multipleeval}, there exists $D' \in \Aplus$ such
that $N = D'D$. Thus, $f = gD'N^{-1}$ and the first assertion of the
lemma follows.
Fix $i \in \{1, \ldots, m\}$ and let $\hat N_i$ be the multiplicative 
inverse of $N/N_i$ in $\Zplus/N_i \Zplus \subset \Aplus /N_i\Aplus$.
Then $\sres_{z_i}(f)$ is the image of $g\hat N_iD' \in \Aplus 
/N_i\Aplus$. Applying $\ev_{c_i}$, we obtain:
$$\ev_{c_i}(\sres_{z_i}(f)) 
= \ev_{c_i}(g\hat N_i) \circ \ev_{c_i}(D').$$
On the other hand, we deduce from $N = D'D$ that
$\ev_{c_i}(D')\circ \ev_{c_i}(D) = 0$ and so that $\ev_{c_i}(D')$
vanishes on $V_i = \im\:\ev_{c_i}(D)$. 
Therefore $\ev_{c_i}(\sres_{z_i}(f))$ vanishes on $V_i$ as well and
the lemma is proved.
\end{proof}

We consider the $K$-linear map: 
$$\begin{array}{rcl} 
\gamma_{\underline{c}, \underline{V}} : \quad 
\calA & \rightarrow & \Hom_{F}(K/\underline{V}, K) \smallskip \\ 
f & \mapsto & \big(\ev_{c_1}(\sres_{z_1}(f)), \ldots, 
\ev_{c_m}(\sres_{z_m}(f))\big). 
\end{array}$$ 
This mapping depends \emph{a priori} on choices of $z_i$-admissible 
morphisms but it follows from Lemma~\ref{lem:simplepole} that the 
restriction of $\gamma_{\underline{c}, \underline{V}}$ to $\Aplus 
D^{-1}$ is independant from any choice. We notice that its restriction 
to $\calZ$ is also uniquely determined because the $\TS_{z_i}$'s have to 
agree with the usual Taylor expansion on the centre. As a conclusion, 
the values of $\gamma_{\underline{c}, \underline{V}}$ on $\calZ \Aplus 
D^{-1}$ are determined without ambiguity.

\begin{definition}
\label{def:LG}
We set $n = \sum_{i=1}^m \dim_F K/V_i$ and consider a positive
integer $k < n$.
The \emph{linearised Goppa code} attached to the parameters $(k, 
\underline c, \underline V)$ is:
$$\LG(k, \underline{c}, \underline{V}) =
\gamma_{\underline{c}, \underline{V}}\big(\Aplus_{<k}\cdot P\big)$$
where $\Aplus_{<k}$ is the subspace of $\Aplus$ consisting of 
Ore polynomials of degree strictly less than $k$ and where $P \in
\calZ \Aplus D^{-1}$ is defined by:
$$\begin{array}{r@{\hspace{0.5ex}}l@{\qquad}l}
P
 & = D^{-1} & \text{if } \theta = \id_K
  \smallskip \\
 & = Z(X)^{-m-1} (X{+}a)^{n-k} D^{-1}
 & \text{if } \theta \neq \id_K \text{ and } \delta = a\delta_0.
\end{array}$$
\end{definition}

We now want to relate linearized Goppa codes to linearized Reed-Solomon 
codes. As in the proof of Lemma~\ref{lem:simplepole}, we pick a Ore
polynomial $D'$ such that $D'D = DD' = N$ and define $\hat 
N_i$ as the multiplicative inverse of $N/N_i$ in $\Zplus/N_i\Zplus$.
We set $\tilde \tau_i = \varepsilon_{c_i}(P N_i)$ where $P$ is the Ore 
polynomial of Definition~\ref{def:LG}; note that $\tilde \tau_i$ is 
well-defined because $P N_i$ has no pole at $z_i = \upsilon(c_i)$.
Noticing that $N_i \equiv D \hat N_i D' \pmod{N_i^2}$, we find 
$\tilde \tau_i = \varepsilon_{c_i}(PD) \circ \varepsilon_{c_i}(\hat N_i D')$,
from what we deduce that $\tilde \tau_i$ vanishes on $V_i = 
\im\:\ev_{c_i}(D)$. 
Therefore $\tilde \tau_i$ induces a surjective $F$-linear morphism 
$\tau_i : K/V_i \to W_i$ where $W_i = \im \tilde\tau_i$.

\begin{lemma}
For all $i \in \{1, \ldots, m\}$, the morphism $\tau_i$ is an
isomorphism.
\end{lemma}

\begin{proof}
We have to show that $\ker \tilde \tau_i = V_i$.
Since $PD$ and $\hat N_i$ are invertible in $\Aplus / N_i\Aplus$, it is enough to 
prove that $\ker \varepsilon_{c_i}(D') = V_i$. The inclusion $V_i
\subset \ker \varepsilon_{c_i}(D')$ has been already noticed in the
proof of Lemma~\ref{lem:simplepole}.
On the other hand, the first part of Theorem~\ref{theo:multipleeval} 
shows that:
$$\sum_{i=1}^m \dim_F \ker \varepsilon_{c_i}(D') \leq \deg D'
= \deg N - \deg D = \sum_{i=1}^m \dim_F V_i.$$
The lemma follows by comparing dimensions.
\end{proof}

We now define $\underline W = (W_1, \ldots, W_m)$ and:
$$\begin{array}{rcl}
\Psi: \quad
\Hom_F(\underline{W},K) & \longrightarrow & \Hom_F(K/\underline{V}, K) 
\smallskip \\
(\varphi_1, \dots, \varphi_m) & \mapsto & 
(\varphi_1\circ \tau_1, \dots, \varphi_m \circ \tau_m)
\end{array}$$
Given that the $\tau_i$'s are all isomorphisms, we deduce that $\Psi$ is 
an isomorphism. Moreover, since composing by an isomorphism obviously 
preserves the rank, $\Psi$ preserves the sum-rank weight and the 
sum-rank distance.

\begin{theorem}
\label{theo:isomLGLRS}
With the above notations, the map $\Psi$ induces an isomorphism of codes 
between $\LRS(k,\underline{c}, \underline{W})$ and $\LG(k,\underline{c}, 
\underline{V})$.
\end{theorem}

\begin{proof}
This follows from the relation
$\varepsilon_{c_i}(\sres_{z_i}(f)) = 
\varepsilon_{c_i}(g) \circ \tilde \tau_i$
which holds true for any $f \in \Aplus P$ in all cases.
\end{proof}

\begin{corollary} \label{theo::caractLG}
Let $k$, $\underline c$ and $\underline V$ as in Definition~\ref{def:LG}
and assume that the $\upsilon(c_i)$'s are pairwise distinct. Then:
\begin{itemize}
\renewcommand{\itemsep}{0pt}
\item the length of $LG(k, \underline{c}, \underline{V})$ is $n
= \sum_{i=1}^m \dim_F K/V_i$,
\item its dimension is $k$,
\item its minimal sum-rank distance is $d = n-k+1$.
\end{itemize}
In particular, the code $LG(k, \underline{c}, \underline{V})$ is
MSRD.
\end{corollary}

\begin{proof}
This is a direct consequence of Theorem~\ref{theo:isomLGLRS} and
Theorem~\ref{theo::caracRSL}.
\end{proof}

\subsection{The duality theorem}

We are finally ready to state and prove our main duality theorem.
We recall that, for an unramified element $c \in K$, we have defined
in \S \ref{ssec:dualOre}:
$$\begin{array}{r@{\hspace{0.5ex}}l@{\qquad}l}
c^\vee
 & = -c & \text{if } \theta = \id_K
  \smallskip \\
 & = \frac 1{c+a} - a & 
 \text{if } \theta \neq \id_K \text{ and } \delta = a \delta_0.
\end{array}$$

\begin{theorem}
Let $k$ and $m$ be two positives integers. 
Let $\underline{c} = (c_1, \dots, c_m)$ be a tuple of $m$ unramified 
elements of $K$ such that the $\upsilon(c_i)$'s are pairwise
distinct.
Let $\underline{V} = (V_1, \dots, V_m)$ be a tuple of $F$-linear subspace 
of $K$. We set 
$n = \dim_F \underline{V}$ and we suppose $k \leqslant n$. Then:
$$\LRS(k,\, \underline c,\, \underline{V})^\perp 
\,=\, \LG(n{-}k,\, \underline c^\vee,\, \underline V^\perp)$$
where $\underline c^\vee = (c_1^\vee, \dots, c_m^\vee)$ and
$\underline V^\perp = (V_1^\perp, \dots, V_m^\perp)$.
\end{theorem}

\begin{proof}
Since the dimensions of $\LRS(k,\underline c, \underline V)$ and 
$\LG(n{-}k,\underline c^\vee, \underline V^\perp)$ sum up to $n$,
it is enough to prove that $\langle \underline\varphi, \underline\psi 
\rangle = 0$ for all $\underline \varphi \in \LRS(k,\underline c, 
\underline V)$ and for all $\underline\psi \in \LG(n{-}k, \underline c^\vee, 
\underline V^\perp)$. For simplicity, we write
$\rho = \ev_{\underline c, \underline V}$ (see Definition~\ref{def:LRS})
and $\gamma = \gamma_{\underline c^\vee, \underline V^\perp}$.
We have to prove that $\langle \gamma(f), \rho(g) \rangle = 0$ for 
$f \in \Aplus_{<n-k} P$ and $g \in \Aplus_{<k}$ where $P$ is the Ore
polynomial introduced in Definition~\ref{def:LG} (for the parameters
$\underline c^\vee$ and $\underline V^\perp$).
We compute:
\begin{equation}
\label{eq:pairing}
\langle \gamma(f), \rho(g) \rangle 
 = \sum_{i=1}^m \tr\big(\gamma(f)^\star\circ\rho(g)\big) 
 = \sum_{i=1}^m \tr\big(\varepsilon_{c_i^\vee}(\sres_{\bar z_i}(f))^\star 
   \circ \varepsilon_{c_i}(g)\big)
\end{equation}
where we have set $\bar z_i = \upsilon(c_i^\vee)$. We recall that
$\bar z_i = - z_i$ if $\theta = \id_K$ and $\bar z_i = z_i^{-1}$
otherwise (the assumption that $c_i$ is unramified indicates that
$z_i$ cannot vanish in the latter case).
From Theorems~\ref{theo:evaldual} and~\ref{theo:resdual}, we deduce
that:
$$\varepsilon_{c_i^\vee}\big(\sres_{\bar z_i}(f)\big)^\star
= \varepsilon_{c_i}\big(\sres_{\bar z_i}(f)^\star\big)
= \varepsilon_{c_i}\big(u \cdot \sres_{z_i}(f^\star)\big)$$
where the multiplicative prefactor $u$ is $-1$ when $\theta =
\id_K$ and $-Z(X)^2$ otherwise.
On the other hand, since $f^\star$ has at most a simple pole at $z_i$ 
and $u$ has no pole at $z_i$, we have
$\sres_{z_i}(u \cdot f^\star) = u \cdot \sres_{z_i}(f^\star)$. Therefore,
we conclude that:
$$\varepsilon_{c_i^\vee}\big(\sres_{\bar z_i}(f)\big)^\star
= \varepsilon_{c_i}\big(\sres_{z_i}(u f^\star)\big)$$
and plugging this into Eq.~\eqref{eq:pairing}, we obtain:
\begin{align*}
\langle \gamma(f), \rho(g) \rangle 
 & = \sum_{i=1}^m \tr\big(\varepsilon_{c_i}(\sres_{z_i}(uf^\star))\circ \varepsilon_{c_i}(g)\big)\\
 & = \sum_{i=1}^m \tr\big(\varepsilon_{c_i}(\sres_{z_i}(uf^\star) \cdot g)\big) 
   = \sum_{i=1}^m \tr\big(\varepsilon_{c_i}(\sres_{z_i}(uf^\star g)\big).
\end{align*}
since again $f^\star$ has at most a simple pole at $z_i$ and $g$ has 
no pole at $z_i$. Now using Theorem~\ref{theo:trdeval}, we end up with:
$$\langle \gamma(f), \rho(g) \rangle 
= \sum_{i=1}^m \Trd\big(\sres_{z_i}(uf^\star g)\big).$$
We write $f = f_0 P$ where $f_0$ is a Ore polynomial of degree
at most $n{-}k{-}1$ and we consider $D'$ such as $D'D = N$. We have:
$$\begin{array}{r@{\hspace{0.5ex}}l@{\qquad}l}
u f^\star g = u P^\star f_0^\star g
 & = \displaystyle \frac{(D')^\star \cdot f_0^\star \cdot g}{N^\star}$$
 & \text{if } \theta = \id_K
  \medskip \\
 & = \displaystyle \frac{(D')^\star \cdot Z(X)^{m-1} \cdot X^{-k} \cdot f_0^\star \cdot g}{N^\star}$$
 & \text{otherwise.}
\end{array}$$
When $\theta = \id_K$, the numerator 
$(D')^\star {\cdot} f_0^\star {\cdot} g$
is a Ore polynomial of degree at most $ms{-}2$ (where we recall that
$s = [K:F]$) and it then follows from the skew residue formula
(Theorem~\ref{theo:residueformula}) that $\langle \gamma(f), \rho(g) \rangle$
vanishes. Similarly, if $\theta \neq \id_K$, we find that the numerator
has only terms in $(X{+}a)^i$ with $i$ in the range $(-s,(m{-}1)s)$ and deduce
from this that the skew residues of $u f^\star g$ at $0$ and $\infty$
both vanish. Hence the skew residue formula (see \cite[Theorem~3.2.1]
{caruso}) also implies the vanishing of $\langle \gamma(f), \rho(g) 
\rangle$ in this case.
\end{proof}

\end{document}